\documentclass[acmsmall]{acmart}
\usepackage{subfigure}
\usepackage{graphicx}
\usepackage{xcolor}
\usepackage{tabularx} 
\usepackage{booktabs}
\usepackage{multirow}
\usepackage{hyperref}
\usepackage{algorithm}
\usepackage{algorithmicx}
\usepackage[noend]{algpseudocode}
\usepackage{amsmath}
\algrenewcommand{\algorithmiccomment}[1]{\hskip0.5em$\triangleright$ \textcolor{darkpurple}{\footnotesize \textit{#1}}}
\algnewcommand\algorithmicforeach{\textbf{for each}}
\algdef{S}[FOR]{ForEach}[1]{\algorithmicforeach\ #1\ \algorithmicdo}

\definecolor{MR}{HTML}{000000} 
\definecolor{R1}{HTML}{000000} 
\definecolor{R2}{HTML}{000000} 
\definecolor{R3}{HTML}{000000} 
\definecolor{R4}{HTML}{000000} 
\definecolor{new}{HTML}{FF0000} 

\usepackage{tikz}

\usepackage{amsthm}
\usepackage{amsopn}
\usepackage{empheq}
\usepackage{mathtools}

\usepackage{color}
\usepackage{graphicx}
\usepackage{listings}
\usepackage{rotating}
\usepackage{tikz}
\usepackage{xspace}
\usepackage{colortbl}
\usepackage{graphicx}
\usepackage{balance}
\usepackage{url}

\usepackage{breakurl}
\usepackage{paralist}
\usepackage[misc]{ifsym}
\usepackage[T1]{fontenc}

\usepackage{stmaryrd}

\usetikzlibrary{decorations,arrows,shapes,shadows}
\usepackage{environ}
\usepackage{etoolbox}
\usepackage{cleveref}

\usepackage{lipsum}
\usepackage{enumitem}
\setlist[itemize]{leftmargin=*}

\usepackage{todonotes}

\Crefname{exam}{Example}{Ex.}
\Crefname{figure}{Figure}{Fig.}
\Crefname{section}{Section}{Sec.}
\Crefname{defi}{Definition}{Def.}
\Crefname{theo}{Theorem}{Thm.}

\newcommand{\mypara}[1]{\noindent\textbf{#1.}}

\newcommand{\query}{Q}

\newcommand{\rel}[1]{\ensuremath{\mathtt{#1}}}

\usepackage[yyyymmdd,hhmmss]{datetime}

\newcommand{\glo}{\ensuremath{\mathcal{G}}}
\newcommand{\loc}{\ensuremath{\mathcal{L}}}

\newcommand{\prov}{\ensuremath{\mathcal{P}}}

\newcommand{\syst}{\textsc{ProvQML}}
\newcommand{\raSQL}{RaSQL}
\newcommand{\mlquery}{Q_{ML}}
\newcommand{\bgdlr}{Q_{bgdlr}}
\newcommand{\bgdlrProv}{Q'_{bgdlr}}

\newcommand{\trainv}{\ensuremath{\rel{Trainv}}}
\newcommand{\trainRel}{\mathcal{T}}

\newcommand{\modelr}{\ensuremath{\rel{Model}}}
\newcommand{\modelprov}{\ensuremath{\rel{ModelProv}}}
\newcommand{\modelRel}{\mathcal{M}}

\newcommand{\predRel}{\ensuremath{\rel{Predict}}}

\newcommand{\gradRel}{\ensuremath{\rel{Gradient}}}
\newcommand{\recurRel}{R}
\newcommand{\recurView}{\rel{VIEW1}}
\newcommand{\localProv}{\prov_{\loc}}
\newcommand{\globalProv}{\prov_{\glo}}

\newcommand{\predInst}{Pred_{i}}
\newcommand{\gradInst}{Grad_{i}}
\newcommand{\grad}{Grad}
\newcommand{\predFunc}{f}
\newcommand{\gradFunc}{g}
\newcommand{\modelFunc}{m}
\newcommand{\mutual}{\mathcal{MT}}
\newcommand{\nonlinear}{\mathcal{NL}}

\definecolor{DarkGreen}{rgb}{0,0.45,0}
\definecolor{DarkRed}{rgb}{0.8,0,0}
\definecolor{DarkYellow}{rgb}{0.6,0.6,0}
\definecolor{DarkGray}{rgb}{0.2,0.2,0.2}
\definecolor{lightgray}{gray}{0.95}

\newcommand{\aggregation}{\gamma}

\newcommand{\join}{\bowtie}

\usepackage{multirow}

\algrenewcommand\algorithmicindent{0.8em}

\algrenewcommand{\algorithmiccomment}[1]{\hskip0.5em$\triangleright$ \textcolor{darkpurple}{\footnotesize \textit{#1}}}

\newtheorem{Theorem}{Theorem}
\newtheorem{Definition}{Definition}

\newtheorem{Example}{Example}

\newenvironment{proof}{\textsc{Proof.}}{\hfill$\square$}

\definecolor{black}{rgb}{0,0,0}
\definecolor{grey}{rgb}{0.8,0.8,0.8}
\definecolor{red}{rgb}{1,0,0}
\definecolor{green}{rgb}{0,1,0}
\definecolor{darkgreen}{rgb}{0,0.5,0}
\definecolor{darkpurple}{rgb}{0.5,0,0.5}
\definecolor{darkdarkpurple}{rgb}{0.3,0,0.3}
\definecolor{blue}{rgb}{0,0,1}
\definecolor{shadegreen}{rgb}{0.95,1,0.95}
\definecolor{shadeblue}{rgb}{0.95,0.95,1}
\definecolor{shadered}{rgb}{1,0.85,0.85}
\definecolor{shadegrey}{rgb}{0.85,0.85,0.85}
\definecolor{oddRowGrey}{rgb}{0.80,0.80,0.80}
\definecolor{evenRowGrey}{rgb}{0.85,0.85,0.85}
\definecolor{wonColor}{rgb}{0.63,0.79,0.95}
\definecolor{lostColor}{rgb}{1.0,0.55,0.35}

\lstdefinestyle{psql}
{
tabsize=2,
basicstyle=\small\upshape\ttfamily,
language=SQL,
morekeywords={PROVENANCE,BASERELATION,PRECEDING,CURRENT,ROW,INFLUENCE,PARTITION,OVER,COPY,ON,TRANSPROV,TRANSSQL,TRANSXML,CONTRIBUTION,COMPLETE,TRANSITIVE,NONTRANSITIVE,EXPLAIN,SQLTEXT,GRAPH,IS,ANNOT,THIS,XSLT,MAPPROV,cxpath,OF,TRANSACTION,SERIALIZABLE,COMMITTED,INSERT,INTO,WITH,SCN,UPDATED, RECURSIVE},
extendedchars=false,
keywordstyle=\bfseries,
mathescape=true,
escapechar=@,
sensitive=true
}

\definecolor{lstpurple}{rgb}{0.5,0,0.5}
\definecolor{lstred}{rgb}{1,0,0}
\definecolor{lstreddark}{rgb}{0.7,0,0}
\definecolor{lstredl}{rgb}{0.64,0.08,0.08}
\definecolor{lstmildblue}{rgb}{0.66,0.72,0.78}
\definecolor{lstblue}{rgb}{0,0,1}
\definecolor{lstmildgreen}{rgb}{0.42,0.53,0.39}
\definecolor{lstgreen}{rgb}{0,0.5,0}
\definecolor{lstorangedark}{rgb}{0.6,0.3,0}
\definecolor{lstorange}{rgb}{0.75,0.52,0.005}
\definecolor{lstorangelight}{rgb}{0.89,0.81,0.67}
\definecolor{lstbeige}{rgb}{0.90,0.86,0.45}

\lstdefinestyle{psqlcolor}
{
tabsize=2,
basicstyle=\small\upshape\ttfamily,
language=SQL,
morekeywords={PROVENANCE,BASERELATION,UNBOUNDED,FOLLOWING,PRECEDING,CURRENT,ROW,PARTITION,OVER,INFLUENCE,COPY,ON,TRANSPROV,TRANSSQL,TRANSXML,CONTRIBUTION,COMPLETE,TRANSITIVE,NONTRANSITIVE,EXPLAIN,SQLTEXT,GRAPH,IS,ANNOT,THIS,XSLT,MAPPROV,cxpath,OF,TRANSACTION,SERIALIZABLE,COMMITTED,INSERT,INTO, WITH, SCN,UPDATED, UNNEST, ASCII, COALESCE, RECURSIVE, as},
extendedchars=false,
keywordstyle=\bfseries\color{lstpurple},
deletekeywords={count,min,max,avg,sum},
keywords=[2]{count,min,max,avg,sum,array_agg,row_number,first_value,last_value, cube, case, when, ascii, unnest, cast, regr, regr_r2, coalesce, last, exp},
keywordstyle=[2]\color{lstblue},
stringstyle=\color{lstreddark},
commentstyle=\color{lstgreen},
mathescape=true,
escapechar=@,
sensitive=true
}

\newbool{Techreport}
\newrobustcmd{\ifnottechreport}[1]{\ifbool{Techreport}{}{#1}}
\newrobustcmd{\iftechreport}[1]{\ifbool{Techreport}{#1}{}}

\booltrue{Techreport}

\AtBeginDocument{%
  }

\acmConference
\acmBooktitle
\setcopyright{cc}
\setcctype{by}
\acmJournal{PACMMOD}

\begin{document}

\title{\syst: 
Efficient Provenance Capture for Machine Learning as Recursive Aggregate SQL [Extended Version]
}

\author{Nkechi J. Akinwale}
\affiliation{%
  \institution{University of Cincinnati}
  \country{USA}
}
\email{akwarinj@mail.uc.edu}

\author{Seokki Lee}
\affiliation{%
  \institution{University of Cincinnati}
  \country{USA}
}
\email{lee5sk@ucmail.uc.edu}


\begin{abstract}

Expressing ML algorithms as recursive aggregate SQL queries over distributed systems enables large-scale declarative data analytics, yet no existing system supports scalable provenance capture, necessary for ML explainability. 
We present $\syst$, a provenance-aware recursive aggregate SQL engine that efficiently captures fine-grained provenance while evaluating ML algorithms as recursive aggregate queries. 
We capture two types of provenance: (i) local provenance that records the contribution of each training instance to gradient computation at a
specific model state, and (ii) global provenance that encodes the sequence of model parameter updates across the iterations.
Extensive experiments on widely used ML datasets demonstrate that our approach efficiently captures both provenance for various ML models over widely used ML datasets.
\end{abstract}

\ccsdesc[500]{Information systems~Relational database query languages}

\keywords{Provenance, Declarative Machine Learning, Aggregate in Recursion}


\maketitle
\section{Introduction}
\label{sec:intro}

Expressing ML algorithms as recursive aggregate SQL queries over distributed systems such as Apache Spark has gained traction 
\cite{SX25,WW21}. For iterative ML tasks, aggregation within recursion is essential for fixpoint convergence~\cite{ZY17}. 
There is also an increasing need for ML explainability~\cite{LL17,KV20,S24}.
In the context of ML as a declarative query, explainability of ML aligns with the database community's long-standing work on explaining query results~\cite{BK01,MG10,AF-GProM18,LL18,LL20,DL21,AL23}. These studies rely on provenance, the record of how a query result is derived from input tuples, which is a well-studied method for such explanations. For ML queries,~\cite{TA23,AL25} envision that provenance can trace how training instances drive model parameter updates, supporting model validation and debugging. It also enables the possibility of computing feature importance~\cite{LL17,KV20}, which quantifies each feature's statistical contribution to the trained model. Despite this, no existing system supports scalable provenance capture for recursive aggregate SQL that is needed for ML explainability.

\begin{figure}[t]
\vspace{-1mm}
\small
\begin{tabular}{@{}c@{\hspace{-0.3cm}}c@{\hspace{0.9cm}}c@{}}
\parbox[c]{0.53\columnwidth}{
\centering
\setlength{\tabcolsep}{6.5pt} 
\textbf{Input training data $\trainRel$} \\
\begin{tabular}{|l|l|l|l|}
\hline
\cellcolor{black}\textcolor{white}{I} &
\cellcolor{black}\textcolor{white}{C} &
\cellcolor{black}\textcolor{white}{V} &
\cellcolor{black}\textcolor{white}{Y} \\
\hline
1   & Age & 46   & 1 \\ 
1   & CWork & 1    & 1 \\
1   & OCat & 1760 & 1 \\
1   & Education & 24   & 1 \\
1   & Sex & 1    & 1 \\
1   & Race & 1      & 1 \\
2   & Age & 45   & 0 \\
2   & CWork & 6    & 0 \\ [-2mm]
\vdots & \vdots & \vdots & \vdots \\
\hline
\end{tabular}
}
&
\parbox[c]{0.4\columnwidth}{%
\centering
\setlength{\tabcolsep}{7pt}
\textbf{Optimized weights} \\
\begin{tabular}{|l|l|l|}
\hline
\cellcolor{black}\textcolor{white}{C} &
\cellcolor{black}\textcolor{white}{P} &
\cellcolor{black}\textcolor{white}{J} \\
\hline
Age &  15.4500 & 100 \\
CWork & -0.4633 & 100 \\
OCat & -24.1800 & 100 \\
Education & 11.5067 & 100 \\
Sex & 0.4367 & 100 \\
Race & -0.4067 & 100 \\
\hline
\end{tabular}%
}
\end{tabular}
\vspace{-4mm}
\caption{ACSIncome data $\trainRel$~\cite{ACSincomedataset21} (verticalized and pre-processed), and trained model parameters (results) over 100 iterations (\texttt{J} = 100).
}
\label{fig:running-example}
\end{figure}


\begin{figure*}[t]
    \centering \vspace{-3mm}
    \includegraphics[width=1.04\textwidth, trim=0 55 0 0, clip]{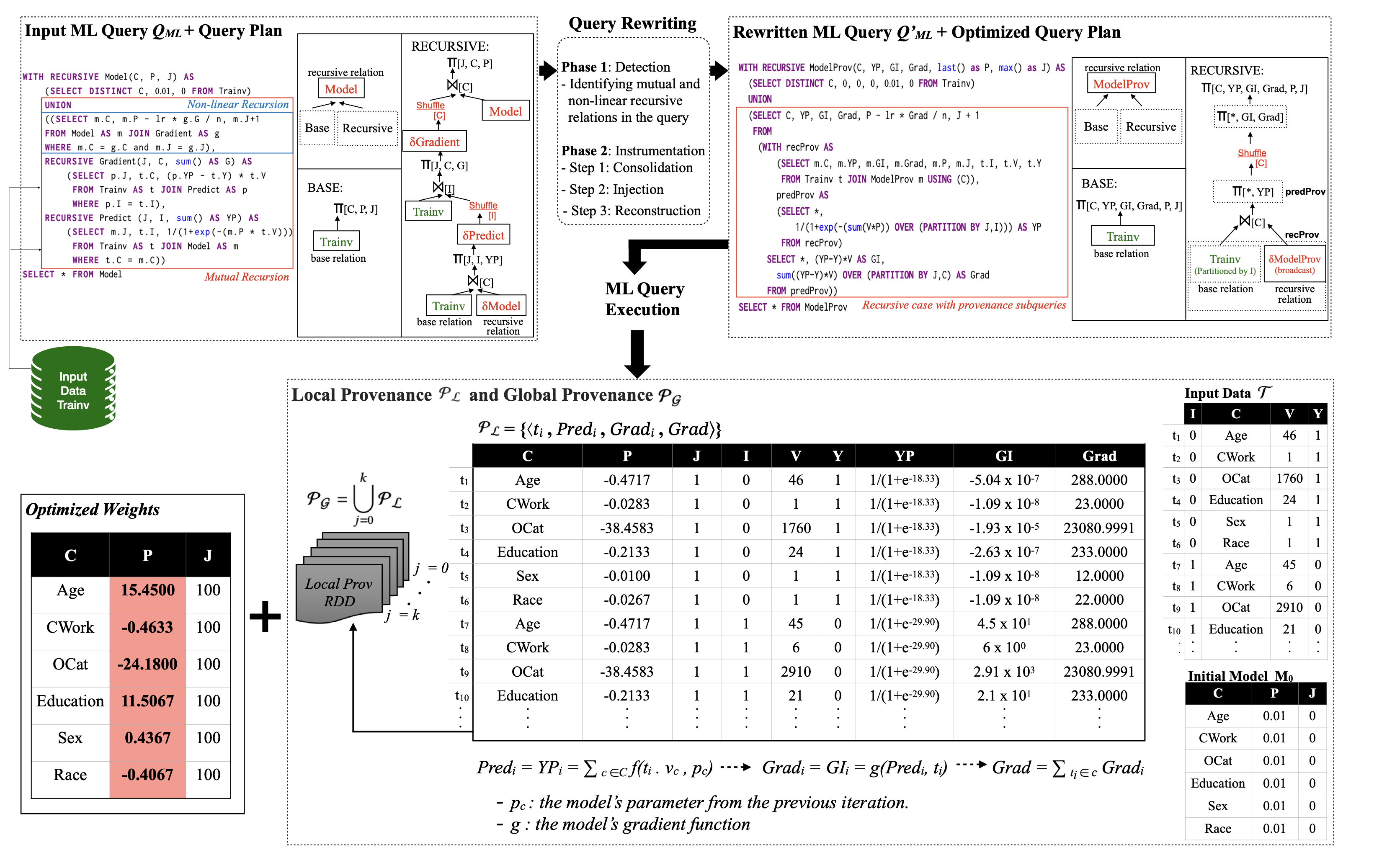}
    \vspace{-7mm}
    \caption{Overview of $\syst$ query evaluation for the running example}
    \label{fig:overview-of-provqml-query-execution}
\end{figure*}

\begin{Example}[Motivating Example]
\label{ex:running-example}
\Cref{fig:running-example} shows a sample of the ACS Income dataset,\footnote{\url{https://fairlearn.org/main/user_guide/datasets/acs_income}} which records demographic and employment attributes, including class of work (\texttt{CWork}) and occupation category (\texttt{OCat}) along with a binary income label (\texttt{Y} = 1 if annual income exceeds \$50{,}000, 0 otherwise), which is reinterpreted as a loan approval outcome. 
The input data $\trainRel$ is stored in a verticalized format where each feature and its value of an instance appear as a tuple.
We will revisit the verticalization in~\Cref{sec:prelim}. 
Tom, an analyst at a bank, uses logistic regression model to determine the loan approval outcomes, and the optimized model's weights are shown in the figure.
His manager raises a concern that loan approvals from the first half the year 
disproportionately favor a particular racial group. 
The model weights alone are insufficient to diagnose the cause, e.g., Race has one of the lowest weights, so Tom needs to trace which training instances drove the weight updates at each iteration.
\end{Example}

In this paper, we present $\syst$, a provenance-aware recursive aggregate SQL engine built on top of Apache Spark that instruments input queries expressing ML algorithms, e.g.,~\cite{WW21}, to efficiently capture fine-grained provenance while evaluating a broad range of such algorithms. 
Unlike~\cite{WW21}, which requires mutual and non-linear recursion with multiple aggregates, $\syst$ uses standard linear recursion with window functions, requiring no new recursion constructs, that can be evaluated in standard (distributed) database systems with a minimal extension. 
The correctness of $\syst$'s query rewriting and optimization relies on two properties established by prior work~\cite{ZY17,GW19,WW21}. The Pre-Mappability (PreM) property~\cite{ZY17,GW19} gives conditions under which extrema aggregates (\texttt{min}, \texttt{max}) can be pushed into recursion while producing the same result as the equivalent aggregate-stratified program. For non-extrema aggregates (\texttt{sum}, \texttt{avg}), the Pre-Countable Cardinality (PCC) property~\cite{WW21} guarantees correctness when the delta relation's cardinality is constant across iterations and pre-computable which is a condition ML workloads such as batch gradient descent satisfy naturally. \Cref{fig:overview-of-provqml-query-execution} gives an overview of $\syst$'s query evaluation pipeline.

\mypara{Provenance capture via query rewriting}
The top of \Cref{fig:overview-of-provqml-query-execution} shows $\syst$ instruments the input ML query $\mlquery$ expressing Tom's logistic regression model in SQL (\Cref{ex:running-example}) into a rewritten query $\mlquery'$ computing the optimized model's weights and their provenance.
We will explain the details of $\mlquery$ and $\mlquery'$ in~\Cref{sec:prelim} and~\Cref{sec:provqml}, respectively.
The formulation of ML queries based on prior work~\cite{WW21} involves mutual and non-linear recursion (highlighted on $\mlquery$), which makes direct provenance instrumentation prohibitively expensive. It increases the number of recursions as provenance capture for aggregate query requires to join with its input relations. The \syst's instrumentation consists of two phases: (i) detecting mutually recursive and non-linearly recursive relations on $\mlquery$, then (ii) collapsing them into a single linear recursive relation whose recursive case expresses intermediate computations as a sequential chain of provenance subqueries using window functions. PreM ensures that aggregates in the rewritten recursive case can be pushed into the recursion without altering the fixpoint, and PCC ensures $\mlquery'$ converges to the same result as $\mlquery$. 
The rewritten query captures two types of provenance, as shown in the bottom half of \Cref{fig:overview-of-provqml-query-execution}. Local provenance $\localProv$ records each training instance's contribution to gradient computation at a specific model state
\textcolor{R3}{by capturing the per-instance prediction $Pred_i$, the per-instance gradient contribution $Grad_i$, and the aggregated gradient $Grad$.} 
\textcolor{R3}{
For example, for the first customer 
Age = 46 annotated with $t_1$, $\syst$ computes provenance for the per-instance prediction $Pred_1 = YP_1 = 1/(1 + \exp(-(sum(t_1.V \cdot p_c, t_2.V \cdot p_c, \cdots, t_6.V \cdot p_c)))$ using the model parameter $p_c$ from the previous iteration.
At $j$ = 1, $p_c$ refers to the initialized model parameter $0.01$, which gives $Pred_1 = YP_1$ = $1/(1 + \exp(-(sum(46,1,1760,24,1,1) \cdot 0.01)))$ = $1/(1+e^{-18.33})$. 
It further computes the corresponding per-instance gradient contribution $Grad_1 = GI_1 = (YP_1 - t_1.Y) \cdot t_1.V$ = $-5.04 \times 10^{-7}$, which contributes to the aggregated gradient value of $288 = Grad_1 + Grad_7 + \cdots$. 
Thus, at $j$ = 1, the updated model parameter $P$ for Age becomes -0.4717.} 
%
%
%
Global provenance $\globalProv$ preserves the sequence of model parameter updates across the fixpoint iterations by accumulating $\localProv$ at each iteration $j = 0, \ldots, k$ as a set of $\localProv$ in RDDs.
The captured provenance enables Tom to reconstruct the full optimization trajectory and trace which training instances drove the weight of, e.g., 
\textcolor{R3}{\texttt{Age} = $15.4500$},
after 100 iterations.

\mypara{Optimization for Provenance Capture}
Capturing local and global provenance introduces two challenges, both of which are reflected in the optimized query plan shown in \Cref{fig:overview-of-provqml-query-execution}. First, the window functions required for update of model's parameters involve sorting and shuffling across worker nodes at every iteration. $\syst$ addresses this by broadcasting the small delta model relation $\delta\modelRel$ to all workers, eliminating the join shuffle entirely, and by partitioning the input training data $\trainRel$ by instance identifier \texttt{I} once upfront. Since the prediction window function partitions on \texttt{I}, this upfront partitioning eliminates its reshuffle; only a single shuffle per iteration remains for the gradient window function, which partitions on \texttt{C}. Second, accumulating global provenance over $k$ iterations incurs $O(|\trainRel| \cdot k)$ memory growth. $\syst$ joins only with $\delta\modelRel_{j-1}$ at each iteration based on the observation that tuples from earlier iterations yield no valid join results due to mismatched iteration keys. Furthermore, it evicts per-iteration RDDs once the accumulated global provenance subsumes the provenance from the previous iterations, keeping the memory usage bounded at $O(|\trainRel|)$ per iteration.

\mypara{Usability of Provenance}
\ifnottechreport{
Provenance-aware recursive training pipelines expose the dependency structure of iterative optimization. By materializing intermediate values such as per-instance predictions ($\predInst$) and per-tuple gradient contributions ($\gradInst$) at each iteration $j$, the system exposes the lineage linking training instances to parameter updates. This enables analysts like Tom to identify influential records and feature-level sensitivity~\cite{TA23,AL25} without retraining the model.
}
\iftechreport{
Provenance-aware recursive training pipelines provide a principled mechanism for exposing the dependency structure of iterative optimization. By materializing intermediate values such as per-instance predictions ($\predInst$) and per-tuple gradient contributions ($\gradInst$) at each iteration $j$, the system makes explicit the lineage linking individual training instances to parameter updates. Instead of observing only the final weight vector, analysts like Tom, can reconstruct how specific instances contribute to the aggregated gradient and how these contributions propagate across recursive state transitions. This fine-grained provenance enables systematic identification of influential records and feature-level sensitivity as indicated in~\cite{TA23,AL25} without re-executing the recursive query entirely (retraining the model). 
}
\iftechreport{
In regulated settings, the resulting provenance can constitute a formally defined mapping from input training instances to final model parameters, supporting auditability of the training process.}

\iftechreport{
Beyond interpretability, embedding provenance within the recursive relation enables computational reuse and adaptive optimization. By persisting iteration indices ($j$), gradient aggregates, and their contributing instances, the system supports warm-start execution, allowing training to resume from a previously materialized iteration state without recomputing earlier steps. Furthermore, the availability of cross-iteration lineage enables convergence diagnostics based on the evolution of gradient magnitudes and directions. Such diagnostics can guide dynamic adjustment of hyperparameters (e.g., learning rate or sampling ratio) within the same declarative framework. By integrating provenance into the recursive query dataflow, the system reduces reliance on external monitoring while enabling selective retraining and incremental refinement grounded in explicit dependency analysis.
}

\mypara{Contributions}
We make the following contributions:
\begin{itemize}
\item We introduce the $\syst$ query language, which extends recursive aggregate SQL with support for subqueries and window functions within the recursive case, enabling provenance capture for a broad range of ML algorithms 
(\Cref{sec:provqml}).

\item We devise query rewriting techniques that 
capture local and global provenance during query evaluation (model training), with correctness guaranteed by PreM and PCC (\Cref{sec:query-rewriting-semantics}).
\item We develop $\syst$, a recursive aggregate SQL engine on top of Apache Spark that executes the rewritten queries with optimizations that reduce shuffle overhead and bound memory growth during provenance capture (\Cref{sec:optimization}).
%
\item We evaluate $\syst$ using 5 ML algorithms and 4 real-world datasets, demonstrating efficient provenance capture 
(\Cref{sec:experiments}).
\end{itemize}

The rest of the paper is organized as follows:~\Cref{sec:prelim} covers basics of recursive aggregate queries including PreM and PCC properties. \Cref{sec:problem-definition} formalizes the problem, and \Cref{sec:provqml} introduces $\syst$ query language. \Cref{sec:query-rewriting-semantics} explains provenance computation and~\Cref{sec:optimization} discusses optimizations. We present experiments in~\Cref{sec:experiments}, related work in~\Cref{sec:relatedwork}, and conclusion in~\Cref{sec:conclusion-future-work}.


\section{Preliminary}
\label{sec:prelim}

\subsection{Expressing ML Algorithm as Queries}
\label{sec:formal-sem}
We revisit the formal semantics~\cite{WW21} that establish how ML models can be expressed declaratively as recursive aggregate queries over data represented in verticalized format. As data volume and feature dimensionality continue to grow, training datasets have become large, high-dimensional, and increasingly sparse. Verticalized data formats provide an efficient way to store such data, e.g., the popular LIBSVM benchmark ~\cite{libsvm}.
\textcolor{R3}{To the best of our knowledge, verticalization is the standard encoding adopted by prior work on declarative ML expressed as recursive queries ~\cite{WW21}, and it is the encoding that makes the ML query expressible as a single, dimension-independent recursive aggregate query.} 
Formally, each training instance $X = \langle I,Y,x_1, \cdots, x_n \rangle$ is transformed into tuples of the form $\langle I,Y,l,v \rangle$, where $l \in [1,...,n]$ denotes the feature dimensions and $v$ is the corresponding feature value. 
Features with value 0 are omitted. 
We use $\trainv(I,C,V,Y)$ to represent a verticalized training data where
$I$ is the identifier of a training instance $inst_i$,
$C$ denotes a feature of $inst_i$, 
$V$ is the value of $C$ for $inst_i$,
and $Y$ denotes the label (target) for $C$ in $inst_i$.


The recursive aggregate queries introduced by~\cite{WW21} over $\trainv$ exhibit two types of recursion. \emph{Mutual recursion} occurs when two or more recursive relations depend on each other, forming a cyclic dependency among relations. 
\emph{Non-linear recursion} occurs when a single recursive step contains multiple recursive relations. 
Both mutual and non-linear recursion introduce additional processing overhead compared to linear recursion, which contains exactly one recursive relation, due to their more complex evaluation paths. 
\cite{WW21} investigated mitigation of the computational overhead of these complex recursions.
\Cref{tab:queryconfig} shows various ML algorithms that differ in prediction, loss, and gradient functions. Note that the general ML query structure remains the same.

\begin{figure}[b]
\hrule
\begin{minipage}{0.95\linewidth}
\centering
\begin{lstlisting}[numbers=left, stepnumber=1, numbersep=5pt, columns=fixed, basewidth=0.6em]
WITH RECURSIVE $\modelr$(C, P, J) AS
    (SELECT DISTINCT C, $0.01$, $0$ FROM $\trainv$)
    UNION
    ((SELECT m.C, m.P - lr * g.G / n, m.J+1
    FROM $\modelr$ AS m JOIN $\gradRel$ AS g
    WHERE m.C = g.C and m.J = g.J),
    RECURSIVE $\gradRel$(J, C, sum() AS G) AS
        (SELECT p.J, t.C, (p.YP - t.Y) * t.V
         FROM $\trainv$ AS t JOIN $\predRel$ AS p 
         WHERE p.I = t.I),
    RECURSIVE $\predRel$ (J, I, sum() AS YP) AS
        (SELECT m.J, t.I, 1/(1+exp(-(m.P * t.V)))        
         FROM $\trainv$ AS t JOIN $\modelr$ AS m
         WHERE t.C = m.C))
SELECT * FROM $\modelr$
\end{lstlisting}
\end{minipage}
\hrule
\vspace{-3mm}
\caption{An example query encoding BGD for LR}
\label{fig:formalsemanticsquery}
\vspace{-3mm}
\end{figure}

\begin{Example}[Example ML query]
\label{ex:ml-q}
~\Cref{fig:formalsemanticsquery} shows an example SQL query that expresses batch gradient descent (BGD) for logistic regression (LR), equivalent to the Datalog version in~\cite{WW21}. 
The query defines three mutually recursive relations: $\modelr$ (lines 1--6), $\gradRel$ (lines 7--10), and $\predRel$ (lines 11--14). The recursive relation \lstinline{$\modelr$(C, P, J)} consists of three columns. \lstinline{J} is the iteration counter, \lstinline{C} denotes a feature dimension, and \lstinline{P} is the model weight for each feature \lstinline{C}. The base case (line 2) initializes the model parameters by selecting distinct features \lstinline{C} from the verticalized training relation $\trainv$ with an initial weight of $0.01$ and iteration counter $0$.
The recursive case (lines 4--13) updates the model weight as \lstinline{m.P - lr * g.G / n}, where \lstinline{lr} is the learning rate and \lstinline{n} is the number of training instances, by joining $\modelr$ with $\gradRel$ on feature dimension \lstinline{C} and iteration \lstinline{J}. $\predRel$ (lines 11--14) joins $\modelr$ with $\trainv$ on \lstinline{C} and computes the per-instance predicted value \lstinline{YP} as sum of \lstinline{1/(1+exp(-(m.P * t.V)))} for each training instance \lstinline{I}. $\gradRel$ (lines 7--10) joins $\predRel$ with $\trainv$ on instance identifier \lstinline{I} and computes the aggregated gradient \lstinline{G} as sum of \lstinline{(p.YP - t.Y) * t.V} for each feature \lstinline{C}. The three relations are evaluated mutually recursively until convergence or \lstinline{J} reaches the maximum iteration, and the final trained model's parameters are returned by the outer \lstinline{SELECT} (line 15).
\end{Example}
%

\setlength{\tabcolsep}{0.2pt} 
\begin{table}[t]
\caption{Query Configuration for Used ML Algorithms}
\vspace{-4mm}
\label{tab:queryconfig}
\centering
\resizebox{\textwidth}{!}{%
\begin{tabular}{p{1cm} p{2.8cm} p{2.7cm} l l l}
\toprule
Query  & ML Algorithm & Predict function $f$ & Loss function $\mathcal{L}$ &
Gradient $g = \nabla_{\theta} \mathcal{L}$ & Regularizer $\lambda$ \\
\midrule
$\query_1$ & Linear regression
& $YP = V \cdot P$
& $(YP - Y)^2$
& $2\cdot (YP - Y)\cdot V$
& N/A \\
$\query_2$ & Lasso regression
& $YP = V \cdot P$
& $(YP - Y)^2$
& $2\cdot (YP - Y)\cdot V$
& $\lambda \cdot \mathrm{sign}(P)$ \\

$\query_3$ & Ridge regression
& $YP = V \cdot P$
& $(YP - Y)^2$
& $2\cdot (YP - Y)\cdot V$
& $\lambda \cdot P$ \\
$\query_4$ & Logistic regression
& $YP = \dfrac{1}{1 + e^{-V \cdot P}}$
& $
\begin{cases}
- \log(YP), \text{if } Y = 1 \\
- \log(1 - YP), \text{if } Y = 0
\end{cases}
$
& $(YP - Y)\cdot V$
& N/A \\
$\query_5$ & SVM
& $YP = V \cdot P$
& $\max(0, 1 - Y \cdot Y_P)$
& $\begin{cases}
- Y \cdot V, \text{if } Y \cdot YP < 1 \\
0,            \text{otherwise}
\end{cases}$
& $\lambda \cdot P$ \\
\bottomrule
\end{tabular}
}
\footnotesize {*Note: $YP$ is the predicted output from the model, $P$ is the weight and $V$ is the feature value. $Y$ is the true label. $\lambda$ is a hyperparameter that controls the regularizer weight.
}
\end{table}
\subsection{PreM and PCC Properties}
\label{sec:prem-pcc}

\mypara{PreM} 
The Pre-Mappability (PreM) property~\cite{ZY17,GW19} provides formal conditions under which extrema aggregates, i.e., min and max, can be pushed into recursion while preserving the semantics of the equivalent aggregate-stratified program. Formally, given a relational algebra expression $T(R_1, \ldots, R_m)$ where $R_i$ is the vector of relations, and a constraint $\gamma$, $\gamma$ is PreM to $T$ if $$\gamma(T(R_1, \ldots, R_m)) = \gamma(T(\gamma(R_1), \ldots, \gamma(R_m)))$$ When the PreM property holds, the semi-naive fixpoint computation of the unstratified program produces the same results as the equivalent aggregate-stratified program, providing both a formal declarative semantics and a more efficient evaluation by eliminating the need to defer the aggregate to a post-recursion stratum. However, PreM applies only to extrema aggregates and is insufficient for aggregates such as sum, count, and avg. 

\mypara{PCC} 
To address the limitation of PreM for non-extrema aggregates, the Pre-Countable Cardinality (PCC) property~\cite{WW21} provides formal semantics for the use of sum, count, and avg in recursion. The PCC property holds when the cardinality of the delta relation $\delta R_i$ at each iteration $i$ of the semi-naive fixpoint computation is nonzero, remains constant across all iterations, and can be pre-computed before the fixpoint computation begins. Under this condition, non-monotonic aggregates such as sum can be computed using only monotonic constructs, since the incremental computation of continuous count and sum is monotonic. Consequently, a recursive program satisfying PCC is equivalent to an aggregate-stratified program where the cardinality is pre-computed at a lower stratum, providing its formal declarative semantics. This condition is naturally satisfied by ML workloads such as gradient descent, where the number of training instances processed per iteration is fixed and known ahead of time.

\subsection{Distributed Semi-Naive Evaluation}
\label{sec:dsn}
Distributed semi-naive evaluation (DSN) executes recursive queries using the semi-naive strategy in a distributed environment, with data partitioned across workers and intermediate results redistributed via shuffle operations ~\cite{WW21}.
Semi-naive evaluation relies on delta-based execution: at iteration \lstinline{j}, only newly derived tuples $\delta$ are used to compute results for iteration \lstinline{j+1}, and the evaluation continues until a fixpoint is reached. This significantly reduces the size of intermediate results. 
%
\ifnottechreport{
Distributed semi-naive evaluation proceeds as a fixpoint loop over Map and Reduce stages. In the Map stage, the join between the base relation $B$ and the delta relation $\delta R$ on the specified partition key $K$ is performed locally on each mapper, generating intermediate results allocated to reducers. In the Reduce stage, each reducer computes the new tuples $D$ for the current iteration as the set difference $\delta R' - R$, updates the recursive relation $R$ for the next iteration via union, and emits $D$.
At each iteration $j$, the Map stage produces intermediate results, which are redistributed across reducers via a shuffle exchange on key $K$. The Reduce stage then computes the new delta $\delta R$. The loop terminates when no new tuples are derived, i.e., $\delta R = \emptyset$, at which point the fully computed recursive relation $R$ is returned.
\syst\ leverages DSN for its query evaluation. The full pseudocode is provided in the extended version~\cite{ProvQMLExtended}.
}
\iftechreport{
Algorithm~\ref{alg:dsn} describes the DSN evaluation procedure. It defines two auxiliary functions specifying the Map (lines 3--5) and Reduce (lines 6--10) stages. In the Map stage, the join between the base relation $B$ and the delta relation $\delta R$ on the specified partition key $K$ is performed locally on each mapper, generating intermediate results allocated to reducers (lines 4--5). The function $F$ denotes any additional selection or projection operations required by the 
recursive program. In the Reduce stage, each reducer computes the new tuples $D$ for the current iteration as the set difference $\delta R' - R$ (line 8), updates the recursive relation $R$ for the next iteration via union (line 9), and emits $D$ (line 10).
The main fixpoint loop begins at line 11, where $\delta R$ is initialized to the results of the base case and $R$ is set to $\emptyset$. At each iteration $j$ (line 13), the Map stage produces intermediate results (line 14), which are redistributed across reducers via a shuffle exchange on key $K$ (line 15). The Reduce stage then computes the new delta $\delta R$ (line 16). The loop terminates when no new tuples are derived, i.e., $\delta R = \emptyset$ (line 17), at which point the fully computed recursive relation $R$ is returned (line 18).
\syst\ leverages DSN for its query evaluation.
}


\iftechreport{
\begin{algorithm}[t]
\caption{Distributed Semi-Naive Evaluation}
\label{alg:dsn}
\begin{algorithmic}[1]
\State Input: Base relation $B$, Partition key $K$
\State Output $R$: Recursive relation, $\delta R$: Recursive relation (Delta)

    \State \textbf{MapStage} ($\delta R$, $B$)
    \ForEach{partition pair of ($\delta R$, $B$)}
    \State \textbf{emit} $F(\delta R \bowtie_{(\delta R.K = B.K)} B)$
    \EndFor
    \State \textbf{ReduceStage} ($\delta R'$, $R$)
    \ForEach{partition pair of ($\delta R'$, $R$)}
        \State $D \leftarrow \delta R' - R$
        \State $R \leftarrow \delta R' \cup R$
        \State \textbf{emit} $D$
    \EndFor
    \State $\delta R \leftarrow \text{Results of base case}$, $R \leftarrow \emptyset$

    \Repeat
        \State \qquad $j \leftarrow j + 1$
        \State \qquad $mapOutput \leftarrow \text{MapStage}(\delta R, B)$
        \State \qquad $\delta R' \leftarrow \text{ShuffleExchange}(mapOutput, \text{key} = K)$
        \State  \qquad $\delta R \leftarrow \text{ReduceStage}(\delta R', R)$
    \Until{$\delta R == \emptyset$}
    \State \Return $R$
\end{algorithmic}
\end{algorithm}
}

\iftechreport{
\subsection{Gradient Descent (GD) in ML} 
Training a model requires finding parameters that minimize the prediction error over the training dataset. Let D = ${(X_i, y_i)}_{i=1}^{n}$ denote a training dataset with $n$ tuples, where each tuple consists of a d-dimensional feature vector $X_i$, and a corresponding target value $y_i$. For classification tasks, the target variable $y_i \in \{-1, +1\}$ (e.g., queries $Q_q$ where $q \in \{4,5\}$ in ~\Cref{tab:queryconfig}); for regression tasks, $y_i \in \mathbb{R}$ (e.g., queries $Q_q$ where $q \in \{1,2,3\}$).
The process of training a model can be formalized as follows. Given a function $f (\theta; X)$ that makes predictions, the goal is to find a set of parameters $\theta$ (i.e., model weights $w$) that minimizes the prediction error, as measured by a loss function $\mathcal{L}$: $\theta^* = \arg\min\nolimits_{\theta} \mathcal{L}\big(f(\theta; X), Y\big)$.
%
$\mathcal{L}$ is typically evaluated using gradient descent (GD), a first-order, gradient-based optimizer widely adopted in ML. Gradient represents the rate of change of the loss function $\mathcal{L}$ with respect to the model parameters $\theta$, while descent imply moving in the direction that reduces $\mathcal{L}$ toward its minimum at each iteration $i$. In this paper, we use batch gradient descent (BGD), a variant of GD widely used in relational engines. BGD computes precise gradients over the entire training dataset, updating $\theta$ at each iteration. As shown in ~\Cref{equ:parameterupdate}, BDG starts from an initial model parameter $\theta_0$, with each update step controlled by the learning rate $\eta$. The iterative process continues until a convergence or a termination condition $j$ is met. 
\begin{equation}
\label{equ:parameterupdate}
\theta_{i+1} = \theta_i - \eta \left[ \sum_{(X,y) \in D} \nabla_{\theta} \mathcal{L}\big(f(\theta_i \cdot X), y\big) + \lambda \theta_i \right]
\end{equation}
%
\noindent where $\nabla_{\theta} \mathcal{L}$ is the gradient of the loss function and $\lambda$ denotes the regularization. ~\Cref{tab:queryconfig} captures the GD algorithms executed in the experiments section of this paper.

}



\section{Problem Definition}
\label{sec:problem-definition}

We consider an ML query that encodes an ML algorithm over 
a verticalized training data $\trainRel$. 
A recursive relation $\modelRel$ (such as \modelr\ in~\Cref{fig:formalsemanticsquery}) maintains the parameters with one tuple per feature dimension. 
The goal is to capture provenance that traces how training instances influence model parameter updates through per-instance predictions and gradient computations.

\subsection{ML Query}
An ML query $\mlquery$ is a recursive aggregate query characterized by two distinct types of recursions that interact through aggregations (see~\Cref{sec:formal-sem}). 
$\mlquery$ operates on input training data $\trainRel$ and the initial model parameters in $\modelRel$ defined using $\trainRel$, e.g., base case in~\Cref{fig:formalsemanticsquery}.
$\modelRel_j$ denotes the model state at iteration $j$.

\begin{Definition}[ML Query]
\label{def:ml-query}
An ML query $\mlquery$ over an input data $\trainRel$ is formally expressed as:
\begin{align*}
\predRel \leftarrow \aggregation \left(\modelRel_{j-1} \join \trainRel \right) \\
\gradRel \leftarrow \aggregation \left(\predRel \join \trainRel \right) \\
\mlquery(\trainRel) = \modelr \leftarrow \text{lfp} \left(\modelRel_j \oplus \gradRel \right)
\end{align*}
where $\oplus$ is the parameter update operator, $\aggregation$ represents aggregations within the inner recursive queries, $\text{lfp}$ is the least fixed point function.
\end{Definition}

The outer recursion computes the fixpoint of model weights, ensuring that each parameter update is derived strictly from the model parameters $p$ of the previous state $\modelRel_{j-1}$ to produce the updated $p$ for the current state $\modelRel_j$. 
Within each outer iteration, the query invokes inner recursive relations, $\predRel$ and $\gradRel$. 
$\predRel$ computes per-instance predictions by aggregating $p$ for each instance, while $\gradRel$ evaluates error signals by aggregating residuals for each feature.

\subsection{Provenance for ML}
In the context of ML, provenance is not merely 
source tuples but a functional trace of how training data 
and model states 
interact through aggregations, enabling explainability and debugging for iterative learning. 
We use $\prov$ to denote provenance and distinguish two types of provenance.

\mypara{Local Provenance}
\emph{Local provenance} $\localProv$ captures the dependency between a training instance $t \in \trainRel$ and its contribution to gradient computation at a specific model state $\modelRel_j$. 

\begin{Definition}[Local Provenance]
\label{def:local-prov}
For each training instance $t_i \in \trainRel$ with feature value $t_i.v_c$, we define $\localProv$ as a tuple $t_{L} = \langle t_i, \predInst, \gradInst, \grad \rangle$ where 
\begin{itemize}
    \item $\predInst = \sum_{c \in C} \predFunc(t_i.v_c, p_c)$ represents the per-instance prediction where $\predFunc$ is the model's prediction function, $C$ is the whole set of features, and $p_c$ is the model's parameter from the previous iteration (e.g., $1/(1+exp(-(m.P * t.V))$ in~\Cref{fig:formalsemanticsquery})
    \item $\gradInst = \gradFunc(\predInst, t_i)$ represents the per-instance gradient contribution where $\gradFunc$ is the model's gradient function (e.g., $(p.YP - t.Y) \cdot t.V$) 
    \item $\grad = \sum_{t_i \in c} \gradInst$ is the aggregated gradient over all instances sharing the same feature category $c$.
\end{itemize}
Formally, $\localProv$ is defined as follows:
\begin{align*}
\localProv(\trainRel,\modelRel) = \{ \langle t_i, \mathcal{A}(t_i,p_c) \rangle \mid t_i \in \trainRel \wedge p_c \in \modelRel \}
\end{align*} 
where $\mathcal{A}$ is the augmentation function that computes $\predInst$, $\gradInst$, and $\grad$. 
\end{Definition}

This provenance explains how each tuple in $\trainRel$ influenced the gradient at a given iteration and enables contribution analysis by quantifying how specific subsets of data adjust the model's weight.

\mypara{Global Provenance}
Global provenance $\globalProv$ preserves the evolution of the model parameters across the sequence of fixpoint iterations. 
In the context of a fixpoint iteration $\recurRel_{i+1} \leftarrow \recurRel_i \cup \Delta \recurRel$, where $\recurRel$ is the recursive relation (e.g., $\modelr$), $\globalProv$ maintains the versioning of $\localProv$ indexed by the iteration counter J.

\begin{Definition}[Global Provenance]
\label{def:global-prov}
The global provenance $\globalProv$ is the union of historical states of $\localProv$ such that
$$
\globalProv = \bigcup_{j=0}^{k} \localProv
$$
where $k$ is the convergence threshold or maximum iteration. 
\end{Definition}

$\globalProv$ ensures that for any model parameter $p$ at iteration $j$, there exists a total mapping back to the specific instance-level gradient $\gradInst$ from the previous iteration $j-1$, which allows for the reconstruction of the optimization path.

We now formally define the problem addressed in this paper:
\begin{center}
\fcolorbox{gray}{lightgray}{
\parbox{0.95\linewidth}{
\begin{itemize}[leftmargin=*,nosep]
    \item \textbf{Input:} An ML query $\mlquery$ and input data $\trainRel$
    \item \textbf{Output:} Optimized model weights $p$ while maintaining local and global provenance $\localProv$ and $\globalProv$, respectively
\end{itemize}
}}
\end{center}

Our approach instruments $\mlquery$ into $\mlquery'$, which computes the optimized model parameters $p$ over $\trainRel$ and $\modelRel$ while capturing both types of provenance on-the-fly.
Although this instrumentation reduces query complexity by eliminating mutual and non-linear recursions,
capturing $\globalProv$ remains computationally expensive because maintaining $\localProv$ throughout the entire iteration can significantly increase intermediate result size.
We explain how our proposed framework $\syst$ generates the instrumented query $\mlquery'$ in~\Cref{sec:query-rewriting-semantics} and describe how it mitigates the computational overhead in~\Cref{sec:optimization}.


\subsection{Computational Challenges}
In evaluating $\mlquery$, each iteration computes a new model state by aggregating contributions from all training instances. 
A parameter update at iteration $j$ depends on all training instances in that iteration and the model state from iteration $j-1$, which recursively depends on all prior iterations. 
Thus, the provenance of any final model parameter implicitly includes a transitive closure over all training instances across all iterations. 
Aggregates like \lstinline{sum} collapse values across tuples into a scalar, obscuring individual instance contributions, while recursive computation discards intermediate parameter states unless explicitly materialized. 
Thus, computing fine-grained provenance incurs substantial overhead, with space complexity $O(|\trainRel| \cdot k)$ where $k$ is the iteration depth, bounded by the product of input data size and iteration count.

\ifnottechreport{
\textcolor{MR}{
A naive instrumentation of $\mlquery$ would require an additional join for each recursive relation. For example, consider the query in~\Cref{fig:formalsemanticsquery}. After \texttt{Predict} computes the per-instance predicted value \lstinline{YP} via \lstinline{sum()}, the relation must join back with $\trainv$ and $\modelr$ to compute $\predInst$ collapsed by the aggregate. Similarly, after \texttt{Gradient} computes the aggregated gradient \lstinline{G} for each feature \lstinline{C}, the relation must join back with the relation encoding $\predInst$ to recover the per-instance gradient contributions $\gradInst$. These additional joins must be evaluated within the mutually recursive structure of $\mlquery$, which already requires multiple shuffle operations per iteration, compounding the overhead further and making the provenance computation unmanageable. A full example of such a naive instrumentation is provided in the extended version~\cite{ProvQMLExtended}.
}
}

\iftechreport{
A naive instrumentation of $\mlquery$ would require an additional join for each recursive relation. Consider the query in~\Cref{fig:formalsemanticsquery}. After \texttt{Predict} computes the per-instance predicted value \lstinline{YP} via \lstinline{sum()}, a provenance relation \texttt{PredProv} must join \texttt{Predict} back with $\trainRel$ and $\modelRel$ on instance identifier \lstinline{I} and feature dimension \lstinline{C} to compute the $\predInst$, collapsed by the aggregate. Similarly, after \texttt{Gradient} computes the aggregated gradient \lstinline{G} for each feature \lstinline{C}, a provenance relation \texttt{GradProv} must join back with \texttt{PredProv} on \lstinline{(C, J)} to capture the per-instance gradient contributions $\gradInst$ and the aggregated $\grad$ per feature. These additional joins must be evaluated within the mutually recursive structure of $\mlquery$, which already requires multiple shuffle operations per iteration. Introducing \texttt{PredProv} and \texttt{GradProv} as further recursive relations compounds this overhead: each of the $k$ iterations incurs additional shuffle operations for both relations, and the full recursive relations --- not just their deltas --- must be retained in memory to support the non-linear recursion, making the computational cost unmanageable. The complete naive instrumentation of $Q_{bgdlr}$ is shown in~\Cref{fig:naive-instrumentation} 
which captures an example rewritten query that can be  
generated by naively instrumenting $\bgdlr$ from \Cref{fig:formalsemanticsquery}.
}
%

\begin{figure}[t]
\centering \hspace{2mm}
\begin{minipage}{0.95\linewidth}
\hrule
\centering
\begin{lstlisting}[numbers=left, stepnumber=1, numbersep=5pt, columns=fixed, basewidth=0.6em]
WITH RECURSIVE Model (C, P, J) AS
  (SELECT DISTINCT 0, C, $0.01$ FROM Trainv)
  UNION
  ((SELECT m.C, m.P - lr * g.G / n, m.J+1
    FROM Model AS m JOIN Gradient AS g
    WHERE m.C = g.C AND m.J = g.J),
    RECURSIVE GradProv (C, P, J, I, V, Y, YP, GI, G) AS
      (SELECT pp.C, pp.P, pp.J, pp.I, pp.V, pp.Y, 
          pp.YP, (pp.YP - pp.Y) * pp.V AS GradI, g.G
       FROM PredProv AS pp JOIN Gradient AS g
          ON pp.C = g.C AND pp.J = g.J),
    RECURSIVE Gradient (J, C, sum() AS G) AS
      (SELECT p.J, t.C, (p.YP - t.Y) * t.V
       FROM $\trainv$ AS t JOIN Predict AS p),
    RECURSIVE PredProv (C, P, J, I, V, Y, YP) AS
      (SELECT t.C, m.P, p.J, t.I, t.V, t.Y, p.YP
       FROM Predict AS p JOIN $\trainv$ AS t ON p.I = t.I
          JOIN Model AS m ON t.C = m.C AND p.J = m.J),
    RECURSIVE Predict (J, I, sum() AS YP) AS
      (SELECT m.J, t.I, 1/(1+exp(-(m.P * t.V))
       FROM $\trainv$ AS t JOIN Model AS m
       WHERE t.C = m.C))
SELECT * FROM Model
\end{lstlisting}
\hrule
\vspace{-3mm}
\caption{An example rewritten query, generated by naively instrumenting $\bgdlr$ from \Cref{fig:formalsemanticsquery}}
\label{fig:naive-instrumentation}
\end{minipage}
\end{figure}




\section{The $\syst$ Language}
\label{sec:provqml}



$\syst$ query language extends \raSQL~\cite{GW19}, a recursive aggregate query engine that runs on Apache Spark, which itself extends standard SQL, especially Common Table Expressions (CTEs) to support the use of basic aggregates (i.e., min, max, sum, count) in recursion. 
The \raSQL\ system addresses the limitation of Spark SQL's compiler, i.e., its lack of native support for aggregates in recursive queries, through a two-step compilation process: first recognizing recursive table references in the query to construct a recursive plan, then executing the plan iteratively via a fixpoint operator until convergence or the maximum iteration is reached~\cite{GW19}. 
\raSQL\ focuses solely on recursive aggregate queries for graph problems such as reachability and shortest paths, supporting neither ML tasks nor provenance computation. 
$\syst$ further extends in three key ways that are necessary to efficiently capture provenance for ML queries: support (i) window functions to capture provenance during recursive query evaluation, (ii) subqueries within the recursive query, and (iii) for user-defined functions (UDFs), such as \lstinline{last()}, for the recursive query. 
%
%

While a previous work~\cite{WW21} addresses the ML limitation by demonstrating how ML workloads can be expressed as recursive aggregate queries in Datalog, but it does not support provenance computation. Furthermore, the query construct described therein (e.g., its SQL version in~\Cref{fig:formalsemanticsquery}) involves complex mutual and non-linear recursions with joins over multiple recursive relations, which would incur prohibitively high computational costs for provenance capture. $\syst$ addresses these limitations through the aforementioned extensions. 
\setlength{\textfloatsep}{4pt}
\begin{figure}[t]
\hrule
\begin{minipage}{1\linewidth}
\centering
\begin{lstlisting}[columns=fixed, basewidth=0.6em, lineskip=-1pt]
WITH RECURSIVE VIEW1 (C, Col$_{\predInst}$, Col$_{\gradInst}$, Col$_{\grad}$, Agg$_P$(), Agg$_J$()) AS
  (SQL-expression-base)
  UNION
  (SELECT C, Col$_{\predInst}$, Col$_{\gradInst}$, Col$_{\grad}$, $\modelFunc$(P,lr,Grad,n), J + 1
   FROM (WITH VIEW2 AS 
            (SQL-expression-recProv),
            (SQL-expression-predProv)
        (SQL-expression-gradProv)))
SELECT * FROM VIEW1
\end{lstlisting}
\end{minipage}
\hrule
\vspace{-3mm}
\caption{$\syst$ query syntax} 
\label{fig:recqmlquery}
\end{figure}

\mypara{Syntax} 
\Cref{fig:recqmlquery} shows the $\syst$ query syntax which begins with the \lstinline{WITH RECURSIVE} construct, followed by the recursive view \recurView\ and its schema definition. 
The schema of \recurView\ consists of three groups of columns. First, \lstinline{C} identifies the feature dimension. Second, \lstinline{Col$_{\predInst}$}, \lstinline{Col$_{\gradInst}$}, and \lstinline{Col$_{\grad}$} carry the per-instance provenance --- the predicted value $\predInst$, the per-instance gradient contribution $\gradInst$, and the aggregated gradient $\grad$ (\Cref{def:local-prov}) --- computed from the subqueries within the recursive case and propagated through \recurView\ at each iteration $j$. Third, the aggregate functions \lstinline{Agg$_P$()} and \lstinline{Agg$_J$()} maintain the latest model weight \lstinline{P} and track the current iteration counter \lstinline{J}, respectively.
$\syst$ allows the use of a UDF to retain the most recent model weight at each iteration, eliminating the need for joins over multiple recursive relations to capture provenance.
While the model's weights are typically expected to decrease as optimization progresses, their numeric values are not guaranteed to be monotonic, so supporting a UDF is necessary.
The content of \recurView\ is defined by a \emph{union} of two sub-queries: a base case (\lstinline{SQL-expression-base}), whose \lstinline{FROM} clause does not refer to any recursive relation and initializes the model, and a recursive case, whose \lstinline{FROM} clause refers to a recursive relation. 
The recursive case computes the updated model weight through the model update function \lstinline{$\modelFunc$(P,lr,Grad,n)}, where \lstinline{P} is the current weight, \lstinline{lr} is the learning rate, \lstinline{Grad} is the computed gradient, and \lstinline{n} is the number of training instances, and increments the iteration counter \lstinline{J+1}, with the recursion terminating when \lstinline{J} is reached.
%
The recursive case is expressed as a sub-query containing a sequence of view definitions via a nested \lstinline{WITH} clause. \lstinline{SQL-expression-recProv} joins \recurView\ with the input training relation to pair each training tuple with its corresponding model weight for the current iteration, capturing the provenance of the model weights with respect to the training data.
\lstinline{SQL-expression-predProv} refers to \lstinline{SQL-expression-recProv} and uses a window function to compute the predicted values along with their provenance. Similarly, \lstinline{SQL-expression-gradProv} refers to \lstinline{SQL-expression-predProv} and uses a window function to compute the gradients along with their provenance. 
This sequential view structure is a design choice in $\syst$, which avoids the complex mutual and non-linear recursions with multiple joins over recursive relations. 
This design significantly reduces the computational cost of provenance capture during recursive query evaluation. Furthermore, $\syst$ adopts an implicit group-by rule for the outer recursive relation \recurView\ 
i.e., no explicit \lstinline{GROUP BY} clause is required, as all non-aggregate columns are automatically treated as grouping columns.


\begin{figure}[t]
\hrule
\begin{minipage}{0.95\linewidth}
\centering
\begin{lstlisting}[numbers=left, stepnumber=1, numbersep=5pt, columns=fixed, basewidth=0.6em] 
WITH RECURSIVE ModelProv(C, YP, GI, Grad, last() as P, max() as J) AS  
  (SELECT DISTINCT C, $0$, $0$, $0$, $0.01$, $0$ FROM $\trainv$)
  UNION
  (SELECT C, YP, GI, Grad, P - lr * Grad / n, J + 1
   FROM 
    (WITH recProv AS 
        (SELECT m.C, m.YP, m.GI, m.Grad, m.P, m.J, t.I, t.V, t.Y
         FROM $\trainv$ t JOIN ModelProv m USING (C)),
        predProv AS 
        (SELECT *, 
            1/(1+exp(-(sum(V*P)) OVER (PARTITION BY J,I))) AS YP
         FROM recProv)
      SELECT *, (YP-Y)*V AS GI, 
        sum((YP-Y)*V) OVER (PARTITION BY J,C) AS Grad         
      FROM predProv))
SELECT * FROM ModelProv
\end{lstlisting}
\end{minipage}
\hrule
\vspace{-3mm}
\caption{Example $\syst$ Query $\bgdlrProv$ capturing provenance for Batch Gradient Descent (BGD) for Logistic Regression (LR)}
\label{fig:recqmlquery-bgdlr}
\end{figure}

\mypara{Example $\syst$ Query} 
Consider the example query $\bgdlrProv$ shown in~\Cref{fig:recqmlquery-bgdlr}, which captures provenance for Batch Gradient Descent (BGD) for Logistic Regression (LR). 
The query begins by defining the recursive view \modelprov\ 
whose schema includes the UDF \lstinline{last()} to retain the most recent model weight at each iteration, and the monotonic aggregate function \lstinline{max()} to track the current iteration counter.
Three additional columns \lstinline{YP}, \lstinline{GI}, and \lstinline{Grad} are introduced to carry the provenance computed through the recursive case: \lstinline{YP} stores the per-instance predicted value $\predInst$, \lstinline{GI} stores the per-instance gradient contribution $\gradInst$, and \lstinline{Grad} stores the aggregated gradient $\grad$ per feature, as defined in~\Cref{def:local-prov}.
The base case (line 2) initializes the model parameters by selecting distinct dimensions \lstinline{C} from the verticalized training relation $\trainv$
(see~\Cref{sec:formal-sem}), with an initial weight of \lstinline{0.01}, iteration counter \lstinline{0}, and initialized  values of 0 for \lstinline{YP}, \lstinline{GI} and \lstinline{Grad}.
The recursive case (lines 4--16) computes the updated model weight as \lstinline{P - lr * Grad / n}, increments the iteration counter \lstinline{J+1}, and involves a sequence of sub-queries. \texttt{recProv} (lines 6--8) joins \modelprov\ 
with $\trainv$ on column \lstinline{C}, forming the provenance of the model weights with respect to the training data. \texttt{predProv} (lines 9--12) computes 
$\predInst$ (\Cref{def:local-prov}) 
using a window function \lstinline{1/(1+exp(-(sum(V*P)))} partitioned by \lstinline{J} and \lstinline{I}, to evaluate the logistic regression prediction function and captures the relationship between each training instance and its corresponding prediction, maintaining the result in \lstinline{YP}. The final \lstinline{SELECT} in the recursive case (lines 13--15) computes 
$\gradInst$ as \lstinline{(YP-Y)*V} stored in \lstinline{GI} (\Cref{def:local-prov}) and 
$\grad$ using a window function \lstinline{sum((YP-Y)*V)} partitioned by \lstinline{J} and \lstinline{C}, capturing the provenance of how $\predInst$ contributes to $\grad$, preserving in \lstinline{Grad}. 
\mypara{Challenges of the Extension}
These extensions are non-trivial within \raSQL's fixpoint framework. Window functions in the recursive case require the query compiler to correctly scope them to the intra-step relation \texttt{recProv} rather than the accumulated recursive relation, re-executing the window evaluation fresh at each fixpoint step. Their semantic correctness within recursion also requires explicit justification via the PCC property (\Cref{sec:prem-pcc}), since standard SQL semantics for window functions do not account for fixpoint convergence. The nested 
subqueries require \raSQL's plan recognition pass must correctly classify each table reference in the recursive case as either an inter-step recursive dependency (e.g., \texttt{ModelProv}) or an intra-step non-recursive view (e.g., \texttt{recProv}, \texttt{predProv}, \texttt{gradProv}) resolved fresh at each iteration. Without this distinction, the compiler could incorrectly flag inner view references as recursive or fail to resolve them entirely, neither of which \raSQL's original pass handles. Finally, the UDF \lstinline{last()} is necessary because model weights are non-monotonic and no standard aggregate captures the most recently computed value. Its correctness relies on the iteration counter \texttt{J} acting as a total order over fixpoint steps, ensuring the fixpoint operator selects the highest \texttt{J} tuple per feature dimension at convergence.


\section{Query Rewriting Semantics}
\label{sec:query-rewriting-semantics}

We first explain how $\syst$ instruments an input ML query $\mlquery$ into a provenance capture query $\mlquery'$ that computes local provenance $\localProv$ and global provenance $\globalProv$. 
In short, the instrumentation proceeds in two phases.
\begin{itemize}[leftmargin=*]
    \item \textbf{Detection} (Section~\ref{sec:detection}): 
    Given $\mlquery$, $\syst$ identifies the mutually recursive relations and non-linearly recursive relations in the query that must be instrumented. 
    \item \textbf{Rewriting} (Section~\ref{sec:rewriting}): 
    Based on $\mlquery$ and the detection result, $\syst$ rewrites each identified recursive relation into a semantically equivalent form using subqueries and window functions, and combines the rewritten relations into $\mlquery'$ that captures $\localProv$ and $\globalProv$ during recursive evaluation.
\end{itemize}

\noindent The semantic correctness of $\mlquery'$ with respect to $\mlquery$ is proved in~\Cref{sec:seman-corr}.

\subsection{Detecting Recursive Relations}
\label{sec:detection}
Given an ML query $\mlquery$ with recursive relations $\mathcal{R} = \{R_1, R_2, \ldots, R_m\}$, where the primary recursive relation $R_1$ is defined by a base case $B_1$ and a recursive case $S_1$, and each auxiliary recursive relation $R_i$ ($i > 1$) is defined solely by a recursive case $S_i$ (i.e., the child of the \texttt{UNION} operator that references $\mathcal{R}$), mutual recursion and non-linear recursion can be identified by each $S_i$: (i) $R_i$ is mutually recursive if $S_i$ references any other $R_j \in \mathcal{R}$ ($j \neq i$), and (ii) non-linearly recursive if $S_i$ contains more than one recursive relations such as $R_j$ in addition to $R_i$.
While a dependency graph over $\mathcal{R}$ combined with strongly connected component detection~\cite{tarjan1972,aho1983} would support a more general detection procedure applicable to arbitrary recursive queries~\cite{abiteboul1995,ullman1988}, the fixed syntactic structure of $\mlquery$ (Definition~\ref{def:ml-query}) makes such generality unnecessary as a single-pass scan over each $S_i$ suffices. Detecting both types of recursions is necessary because they determine how provenance is propagated in the subsequent rewriting step (Section~\ref{sec:rewriting}). Note that the detection algorithm can be extended to support more complex ML algorithms, such as neural networks, by generalizing the inspection to account for additional base cases per recursive relation in $\mlquery$ and deeper recursive dependencies which we leave as future work.

\begin{algorithm}[t]
\caption{\textsc{DetectRecursiveRelations}}
\label{alg:detect}
\begin{algorithmic}[1]
\Require ML query $\mlquery$
\Ensure Mutual recursion $\mutual$,
        non-linear recursion $\nonlinear$

\State 
$\mutual \leftarrow \emptyset$,\;
$\nonlinear \leftarrow \emptyset$,\;
$\mathit{refs} \leftarrow \emptyset$

\State
$\mathcal{R} \leftarrow \textsc{ExtractRecursiveRel}(\mlquery)$,\;
    \Comment{Extract recursive relations}

\For{each $R_i \in \mathcal{R}$}
    \State $S_i \leftarrow \textsc{ExtractRecursiveCase}(R_i)$ \Comment{Extract recursive case of $R_i$}
    \State $\mathit{refs} \leftarrow \{R_j \in \mathcal{R} \mid R_j \in S_i,\; j \neq i\}$ \Comment{Scan $S_i$ for cross-references}
    \If{$\mathit{refs} \neq \emptyset$}
        \State $\mutual \leftarrow \mutual \cup \{R_i\}
               \cup \mathit{refs}$
               \Comment{Record mutually recursive relations}
        \If{$R_i \in S_i$}
            \State $\nonlinear \leftarrow \nonlinear \cup \{R_i\}$
               \Comment{More than one recursive relations}
        \EndIf
    \EndIf
\EndFor

\State \Return $(\mutual,\, \nonlinear)$
\end{algorithmic}
\end{algorithm}

Algorithm~\ref{alg:detect} shows the detection procedure. Given $\mlquery$ as input, it first extracts the set of recursive relations $\mathcal{R}$ from the \lstinline{RECURSIVE} blocks (line 2). It then iterates over each $R_i \in \mathcal{R}$ (line 3), extracting its recursive case $S_i$ (line 4) and scanning $S_i$ for cross-references to other relations in $\mathcal{R}$ (line 5), i.e., any $R_j \in \mathcal{R}$ with $j \neq i$. If at least one such cross-reference is found (\texttt{refs} $\neq \emptyset$), all involved relations are recorded as mutually recursive (lines 6--7). If $R_i$ also appears in its own recursive case $S_i$ (line 8), it is additionally recorded as non-linearly recursive (line 9), since $S_i$ then contains more than one recursive relation. The procedure returns the pair $(\mathcal{MT}, \mathcal{NL})$ upon completing the scan (line 10).
\begin{Example}
For the example query encoding BGD for LR in Figure~\ref{fig:formalsemanticsquery}, the procedure yields a single mutually recursive group $\mutual = \{$\modelr$, \texttt{Gradient}, \texttt{Predict}\}$ and identifies $\nonlinear = \{ $\modelr$ \}$ as a non-linear recursive relation, since its recursive case includes $\modelr$, joining with \texttt{Gradient}.    
\end{Example}
The procedure requires a single pass over each recursive case $S_i$, giving an overall complexity of $\mathcal{O}(|\mathcal{R}| \cdot |S|)$ where $|\mathcal{R}|$ is the number of recursive relations in $\mlquery$, and $|S|$ is the number of a recursive case in $R_i \in \mathcal{R}$. In practice, both are small for ML queries, 
so the procedure adds negligible overhead.

\subsection{Rewriting for Provenance Capture}
\label{sec:rewriting}

Given $\mlquery$ and the detection result $(\mutual, \nonlinear)$ produced by Algorithm~\ref{alg:detect}, the rewriting instruments $\mlquery$ into a provenance capture query $\mlquery'$. The key idea is to consolidate the mutually recursive relations in $\mutual$ into a single linear recursive relation and to make the intermediate computations of $\localProv$ explicit as subqueries with window functions within the recursive case. Thus, per-instance predictions $\predInst$ and gradient contributions $\gradInst$ along with their provenance are materialized at each iteration $j$. 
Rather than evaluating intermediate results (i.e., $\predInst$, $\gradInst$, and $\grad$ in~\Cref{def:local-prov}) across multiple mutually recursive relations, they are computed as a sequential chain of subqueries within the recursive case of the single consolidated relation. 
For example, as shown in Figure~\ref{fig:recqmlquery-bgdlr} (lines 6--8), \lstinline{recProv} joins $\trainRel$ with the current model state $\modelRel_j$, providing the instance-level lineage required for $\localProv$, while \lstinline{predictProv} and the gradient computation apply window functions over \lstinline{recProv} to 
compute $\predInst$, $\gradInst$, and $\grad$ without any further joins. This ensures that provenance is captured through a single join per iteration, avoiding the additional join overhead that would arise from instrumenting the original mutually recursive structure in $\mlquery$. 
%
Since all ML algorithms considered in this paper (\Cref{tab:queryconfig}) share the same recursive structure of $\mlquery$ (\Cref{def:ml-query}), with a recursive relation for prediction and another recursive relation for gradient computation,~\Cref{alg:rewrite} is sufficient to instrument all of them uniformly.

Recall that the semantic justification for this transformation relies on the PreM and PCC properties (Section~\ref{sec:prem-pcc}). PreM ensures that the aggregates in the window functions in the rewritten recursive case can be pushed into the recursion while preserving the semantics of the equivalent aggregate-stratified program. PCC ensures that $\mlquery'$ converges to the same fixpoint as $\mlquery$, since the cardinality of the delta relation at each iteration remains constant and can be pre-computed, guaranteeing that the provenance subqueries do not alter the computed model parameters. Thus, PreM and PCC guarantee that the transformation is semantics-preserving with respect to both the output relation and $\localProv$.

\begin{algorithm}[t]
\caption{\textsc{RewriteForProvenance}}
\label{alg:rewrite}
\begin{algorithmic}[1]
\Require $\mlquery$,
         mutual recursion $\mutual$,
         non-linear recursion $\nonlinear$
\Ensure  $\mlquery'$

\State $\mlquery' \leftarrow \mlquery$

\If{$\mutual \neq \emptyset$}
    \Comment{Step 1: Relation consolidation}
    \State $R_{\mathit{head}} \leftarrow \textsc{SelectHead}(\mutual)$
    \State $\textsc{AugmentSchema}(R_{\mathit{head}}
           )$ 
    \Comment{\textit{Add feature column $C$, provenance columns (e.g., $Col_{\predInst}, \mathit{Col}_{\gradInst}, \mathit{Col}_{\grad}$), and aggregates $\mathit{Agg}_P()$ and $\mathit{Agg}_J()$ to $R_{head}$}}
    \For{each $R_i \in \mutual \setminus \{R_{\mathit{head}}\}$}
        \State $\textsc{ConvertToSubquery}(R_i, R_{\mathit{head}}, \mlquery')$
        \Comment{Move $R_i$ into the recursive case of $R_{\mathit{head}}$ in $\mlquery'$}
    \EndFor
\EndIf
\If{$\nonlinear \neq \emptyset$}
    \State $R_{\mathit{pred}} \leftarrow \textsc{GetPredictRelation}(\mutual)$ \Comment{Step 2: Identify Predict}
    \State $\mathit{recProv} \leftarrow \textsc{BuildRecProv}(R_{\mathit{pred}})$ \Comment{Rewrite Predict  as \texttt{recProv}}
    \State $\mathit{predictProv} \leftarrow
               \textsc{BuildPredProv}(\mathit{recProv})$
           \Comment{Evolve \texttt{recProv} into \texttt{predictProv} computing $\predInst$}
    \State $R_{\mathit{grad}} \leftarrow \textsc{GetGradientRelation}(\mutual)$
           \Comment{Step 3: Identify Gradient}
    \State $\mathit{Grad} \leftarrow
           \textsc{BuildGrad}(R_{\mathit{grad}}, \mathit{predictProv})$
           \Comment{Build gradient subquery over \texttt{predictProv} computing $\gradInst$ and $\grad$}
    \State $\textsc{ReplaceJoinClause}(R_{\mathit{head}}, 
        \{\mathit{recProv}, \mathit{predictProv}, \mathit{Grad}\}, \mlquery')$
           \Comment{Replace join clause of $R_{\mathit{head}}$ with the subquery chain}
\EndIf
\State \Return $\mlquery'$
\end{algorithmic}
\end{algorithm}

Algorithm~\ref{alg:rewrite} formalizes the rewriting in three steps.

\begin{itemize}
\item{\textbf{Step 1: Relation consolidation (lines 2--6).}}
When $\mutual \neq \emptyset$, the mutually recursive relations are collapsed into a single recursive relation. The head relation $R_{\mathit{head}}$ is designated as the primary recursive relation in $\mlquery'$, and it is defined with a feature column C, columns for maintaining provenance (e.g., Col$_{\predInst}$, Col$_{\gradInst}$, and Col$_{\grad}$), in addition to
\lstinline{Agg$_P$()} and \lstinline{Agg$_J$()} which carry converged values across fixpoint iterations. For example, in Figure~\ref{fig:recqmlquery-bgdlr}, 
\lstinline{YP}, \lstinline{GI}, and \lstinline{Grad} maintains provenance computed through the recursive case, and \lstinline{last()} retains the most recently computed model weight at each iteration $j$, 
while \lstinline{max()} tracks the current iteration counter, reflecting \syst's vertical aggregate semantics in which aggregate functions are applied column-wise over the verticalized relation $\trainRel$ to derive a single converged value per feature at each iteration. The remaining relations in $\mutual$ are converted to non-recursive subqueries within the recursive case of $R_{\mathit{head}}$ via \textsc{ConvertToSubquery}.

\item{\textbf{Step 2: Provenance subquery injection (lines 7--10).}}
When $\nonlinear \neq \emptyset$, the \texttt{Predict} relation is retrieved from $\mutual$ and processed in two stages. First, \textsc{BuildRecProv} rewrites it as a subquery \texttt{recProv}, preserving all tuple-level attributes required for $\localProv$. Second, \textsc{BuildPredProv} evolves \texttt{recProv} into a second subquery \texttt{predictProv} by applying partitioned window functions to compute the per-instance prediction $\predInst$, replacing the original inline aggregate.

\item{\textbf{Step 3: Gradient reconstruction (lines 11--13).}}
The \texttt{Gradient} relation is retrieved from $\mutual$ and rewritten as a subquery $\mathit{Grad}$ via \textsc{BuildGrad}, defined over \texttt{predictProv} and computing $\gradInst$ and the aggregated $\grad$ using partitioned window functions. This eliminates the recursive dependency on $\gradRel$ and completes the capture of $\localProv$ within a single recursive case. Finally, \textsc{ReplaceJoinClause} substitutes the subquery chain \{$\mathit{recProv}$, $\mathit{predictProv}$, $\mathit{Grad}$\} into the join clause of $R_{\mathit{head}}$, replacing the original mutually recursive join structure.
\end{itemize}

\begin{figure}[t]
    \centering \vspace{-2mm}
    \includegraphics[width=0.95\linewidth, trim= 0 0 0 15]{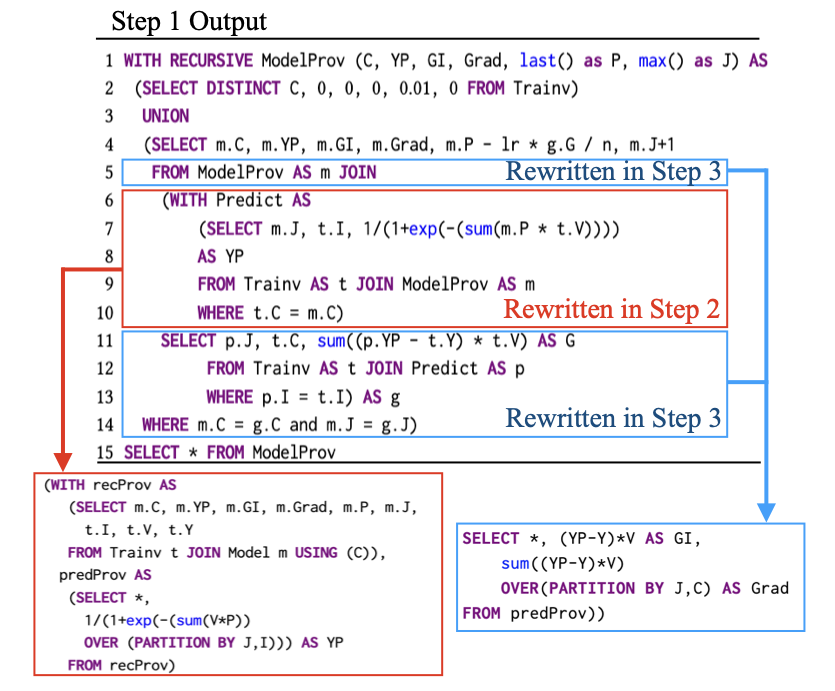}
     \vspace{-6mm}
    \caption{
    \textcolor{MR}{
    Example of intermediate rewritten query for each step of~\Cref{alg:rewrite}, transforming $\bgdlr$ 
    into $\bgdlrProv$ 
    }}
    \label{fig:step-by-step-rewritten-query}
\end{figure}


\begin{Example}
\label{ex:rewrite}
Applying Algorithm~\ref{alg:rewrite} to $\bgdlr$ in~\Cref{fig:formalsemanticsquery} produces \color{MR}{the rewritten query \textcolor{black}{$\bgdlrProv$ in~\Cref{fig:recqmlquery-bgdlr},} through three steps.
The step-wise intermediate rewritten query is shown in~\Cref{fig:step-by-step-rewritten-query}.
In Step 1, \textcolor{black}{$\modelprov$ becomes the head relation (line 1 of Step 1 output), containing  with \lstinline{last() as P}, \lstinline{max() as J},} and provenance columns \lstinline{YP}, \lstinline{GI} and \lstinline{Grad}.
The remaining mutually recursive relations \texttt{Predict} and \texttt{Gradient} in $\mutual$ are converted to non-recursive relations and nested as a subquery within the recursive case of $\:\modelprov$. 
Step 2 rewrites \texttt{Predict} (lines 6--10 of the Step 1 output) into \texttt{recProv} and \texttt{predProv}, \textcolor{black}{creating a single linear join and materializing the per-instance prediction $\predInst$.}
Step 3 reconstructs \texttt{Gradient} (lines 11--13). The join clause (lines 5 and 14) is removed from the recursive case of \modelprov\ and folded into the subquery, yielding $\bgdlrProv$ (\Cref{fig:recqmlquery-bgdlr}), which captures $\localProv$ at each iteration $j$.
}
Since $\localProv$ is materialized at every iteration $j$ through the sequential subquery structure of $\bgdlrProv$, $\globalProv$ is accumulated as the union of these per-iteration snapshots across all fixpoint iterations, indexed by the iteration counter \lstinline{J}.
\end{Example}

\subsection{Semantic Correctness}
\label{sec:seman-corr}
We now state the correctness of the rewritten ML query for computing optimized model weights while maintaining provenance.

\begin{Theorem}
\label{thm:correctness} (Correctness.) 
The rewritten $\mlquery'$, instrumented for input ML query $\mlquery$, returns the optimized model weights defined by $\mlquery$ and the $\localProv$ and $\globalProv$.
\end{Theorem}

\ifnottechreport{
{\color{MR}
\begin{proof}
We provide proof sketch showing that $\mlquery'$ correctly computes the optimized weights and captures complete provenance.
The full proof is shown in ~\cite{ProvQMLExtended}.

\mypara{Correctness of model weights} 
By induction on $j$, the base case is immediate since both $\mlquery$ and $\mlquery'$ 
initialize $\modelRel$ identically in terms of the feature column C, weight P, and iteration counter J (line 2 in~\Cref{fig:formalsemanticsquery,fig:recqmlquery-bgdlr}),
%
so $\modelRel_0$ coincides trivially. For the inductive step, assume $\modelRel_{j-1}$ agrees between the two queries. \texttt{recProv} reconstructs the join $\modelRel_{j-1} \join_C \trainRel$ underlying $\mlquery$'s \texttt{Predict} relation, since $\modelRel$'s aggregate-annotated schema (\lstinline{Agg}$_P$(), \lstinline{Agg}$_J$()) ensures a tuple per feature $C$ represents the current model state at each iteration, so the join picks up the correct, up-to-date parameter value.
\texttt{predProv} and the final windowed \texttt{sum} reproduce $\mlquery$'s \texttt{Predict} and \texttt{Gradient} aggregates exactly, since a window aggregate over a partition computes the same value as the corresponding group-by aggregate, and the weight update itself uses the identical arithmetic expression in both queries. That this per-step equivalence is preserved across the full recursion, i.e., that the rewritten linear recursion converges to the same fixpoint as the original mutual/non-linear one, follows from PreM (safe for the extrema aggregates consolidating $\modelprov$) and PCC (safe for the sum aggregates for \texttt{predProv} and \texttt{Grad}, since the delta relation's cardinality is constant and pre-computable for batch-style ML workloads). Thus, $\modelRel_j$ agrees for all $j \leq k$, and the final weights returned by $\mlquery'$ match $\mlquery$.

\mypara{Provenance completeness} 
The provenance subqueries are non-recursive over $\modelRel$ and $\trainRel$, introducing no spurious tuples, and the join on $C$ with $\modelRel$'s one tuple per feature ensures every $t_i \in \mathcal{T}$ is represented exactly once. Thus, the materialized $\langle \predInst, \gradInst, \grad \rangle$ at each $j$ matches $\localProv$ (Def.~\ref{def:local-prov}) exactly. Since this holds at every iteration and the accumulation of per-iteration snapshots is simply their union, the accumulated result equals $\globalProv$ (Def.~\ref{def:global-prov}) precisely.
\end{proof}
}
}

\iftechreport{
\begin{proof}
We proceed by induction on the iteration index $j$, proving simultaneously (i) that the $\modelRel$ state computed by $\mlquery'$ at step $j$, restricted to the feature, weight, and iteration columns $(C,P,J)$ shared with $\mlquery$'s schema, equals $\modelRel_j$ as computed by $\mlquery$, and (ii) that the provenance materialized by $\mlquery'$ at step $j$ --- carried in the additional columns \lstinline{YP}, \lstinline{GI}, and \lstinline{Grad} introduced by~\Cref{alg:rewrite} --- equals $\localProv$ at $\modelRel_j$ per~\Cref{def:local-prov}.


\mypara{Base case ($j=0$)} Both $\mlquery$ (e.g., \Cref{fig:formalsemanticsquery}, line 2) and $\mlquery'$ (\Cref{fig:recqmlquery-bgdlr}, line 2) initialize the feature column, weight, and iteration counter identically, assigning $C$, $0.01$, and $0$, respectively.
By Step 1 of~\Cref{alg:rewrite} (\textsc{AugmentSchema}), the base case of $\mlquery'$ additionally initializes the provenance columns \lstinline{YP}, \lstinline{GI}, and \lstinline{Grad} to $0$, reflecting that no prediction or gradient computation has yet occurred at $j=0$.
Thus, although $\mlquery'$'s base relation carries a strictly larger schema than $\mlquery$'s, $\modelRel_0$ is syntactically and semantically identical in both queries when restricted to the shared columns $(C,P,J)$. \textsc{AugmentSchema} extends the base case with trivially initialized provenance columns but does not alter the values or semantics of the original feature, weight, and iteration columns. No provenance is claimed at $j=0$ since no gradient step has yet been taken, so (ii) holds.
Having established this schema correspondence, for the remainder of the proof we write $\modelRel_j$ for $\mlquery'$'s state at step $j$ to mean its projection onto $(C,P,J)$, suppressing the additional provenance columns except where part (ii) is directly concerned.

\mypara{Inductive hypothesis} Assume that at step $j-1$, $\modelRel_{j-1}$ computed by $\mlquery'$ equals $\modelRel_{j-1}$ computed by $\mlquery$, and that $\localProv$ at step $j-1$ (if $j-1 \geq 1$) was correctly captured by $\mlquery'$.

\mypara{Inductive step ($j-1 \to j$)} We show the recursive case of $\mlquery'$ (e.g., \Cref{fig:formalsemanticsquery}, lines 4--16) computes the same $\modelRel_j$ as the recursive case of $\mlquery$ (e.g., \Cref{fig:formalsemanticsquery}, lines 4--14), and captures $\localProv$ at $j$ correctly.

\begin{enumerate}
\item \emph{Join equivalence.} \texttt{recProv} (\Cref{fig:formalsemanticsquery}, lines 6--8) computes exactly $\modelRel_{j-1} \join_C \trainRel$, the join consumed by \texttt{Predict} in $\mlquery$ (\Cref{def:ml-query}). 
By Step 1 of~\Cref{alg:rewrite}, \modelprov\ is declared with 
\lstinline{YP}, \lstinline{GI}, \lstinline{Grad}, \lstinline{Agg}$_P$() = \lstinline{last()}, \lstinline{Agg}$_J$() = \lstinline{max()} (\Cref{fig:formalsemanticsquery}, line 1); under \syst's aggregate-in-recursion semantics (\Cref{sec:formal-sem}), any reference to \modelprov --- including within its own recursive case --- is exposed as a single collapsed tuple per feature $C$, obtained by applying these aggregates over all tuples derived for $C$ so far (a logical-view guarantee independent of physical retention; see~\Cref{sec:opt-global}). Since $J$ is a total order over fixpoint steps and each $C$ is updated once per iteration (\Cref{sec:provqml}), \lstinline{max()} and \lstinline{last()} jointly identify the unique tuple from iteration $j-1$, which by the inductive hypothesis is $\modelRel_{j-1}$. The join with $\trainv$ thus computes $\modelRel_{j-1} \join_C \trainRel$, matching $\aggregation(\modelRel_{j-1} \join \trainRel)$ in~\Cref{def:ml-query}.
\item \emph{Prediction equivalence.} \texttt{predProv} computes $YP$ (e.g., \Cref{fig:formalsemanticsquery}, lines 10--12). A window aggregate partitions its input into equivalence classes under $k$ and computes $\texttt{agg}(e)$ over each class independently, differing from the corresponding \texttt{GROUP BY k} aggregate only in that the result is broadcast back to every row of the class rather than collapsed to one row per class; since the same multiset of values within each class is aggregated either way, the two produce identical scalars. Applying this to the windowed sum over \texttt{recProv} shows it returns, for each tuple, the identical scalar that the sum()-annotated \texttt{Predict} relation of $\mlquery$ would compute when grouped by $(J,I)$ over the \emph{same} join input. Thus, the per-instance prediction $Pred_i = YP_i$ (\Cref{def:local-prov}) computed by $\mlquery'$ equals the one implicitly computed inside $\mlquery$'s \texttt{Predict} relation.
\item \emph{Gradient equivalence.} The final \texttt{SELECT} of the recursive case (e.g., \Cref{fig:formalsemanticsquery}, lines 10--12). A window aggregate computes $GI$ per tuple and $\grad$ 
By step 2, $YP$ here equals $\mlquery$'s prediction, so $GI$ equals $Grad_i$ of~\Cref{def:local-prov} exactly as \texttt{Gradient}'s inner select expression, e.g., \texttt{(p.YP - t.Y) * t.V} in $mlquery$ (\Cref{fig:formalsemanticsquery}, line 8). By the same window/group-by equivalence used in step 2, the windowed \lstinline{sum} over $(J,C)$ equals the sum()-annotated aggregate of \texttt{Gradient} in $\mlquery$ grouped by $(J,C)$. Thus, $\grad$ computed by $\mlquery'$ equals $\mathcal{G}$ computed by $\gradRel \leftarrow \aggregation \left(\predRel \join \trainRel \right)$ in $mlquery$.
\item \emph{Weight update equivalence.} Both queries update the weight via the identical arithmetic expression \texttt{P - lr * Grad / n} (resp. \texttt{m.P - lr * g.G / n}) and increment \texttt{J+1}; since the \texttt{Grad}/\texttt{G} operands were shown equal in step 3, and since \texttt{recProv} already carries forward the correct $m.P$ (step 1), the updated weight computed by $\mlquery'$ equals $\modelRel_j$ as computed by $\mlquery$'s outer recursive \modelprov\ relation ($\oplus$ in~\Cref{def:ml-query}).
\item \emph{Fixpoint-level validity of the transformation.} Steps 1--4 above show \emph{per-step} value equivalence; it remains to argue this equivalence is preserved \emph{under recursion}, i.e., that iterating the rewritten recursive case yields the same least fixed point as iterating the original mutually recursive one. This is exactly what PreM and PCC guarantee (\Cref{sec:prem-pcc}): PreM ensures the extrema aggregates \lstinline{Agg}$_P$(), \lstinline{Agg}$_J$() used to consolidate \modelr, \texttt{Predict}, \texttt{Gradient} into a single relation can be pushed inside the recursion without altering the fixpoint relative to the equivalent aggregate-stratified program; PCC ensures the non-extrema sum aggregates used in \texttt{predProv}/\texttt{Grad} are similarly safe to push inside the recursion because the cardinality of the delta relation $\delta\modelr_j$ (one tuple per feature $C$) is constant across iterations and pre-computable ahead of the fixpoint computation --- a condition satisfied by construction for the batch-gradient ML workloads of~\Cref{tab:queryconfig}, where the instance count $n$ and feature count $|C|$ are fixed. Together, PreM and PCC guarantee that the \emph{sequence} of states $\modelRel_0, \modelRel_1, \dots$ produced by $\mlquery'$'s single linear recursion coincides, at every $j$, with the sequence produced by $\mlquery$'s mutual/non-linear recursion, and that both reach the same fixpoint (or terminate at the same maximum iteration $k$).
\end{enumerate}

By steps 1--5 above and the inductive hypothesis, the model weights in $\modelRel_j$ computed by $\mlquery'$ equals the model weights in $\modelRel_j$ computed by $\mlquery$, establishing (i) for all $j \leq k$.

\mypara{(ii) Provenance soundness and completeness at step $j$} We now show the tuple $t_{L} = \langle t_i, \predInst, \\ \gradInst, \grad \rangle$ materialized by $\mlquery'$ for each $t_i \in \trainRel$ matches~\Cref{def:local-prov} exactly, with no omission or duplication.

\begin{itemize}
\item \emph{Soundness (no spurious tuples).} \texttt{recProv}, \texttt{predProv}, and the final gradient subquery are all \emph{non-recursive} subqueries defined purely in terms of $\modelprov $ (already shown correct by (i)) and $\trainv$; they introduce no new base tuples and no join beyond $\modelRel_{j-1} \join_C \trainRel$ (step 1 above). $\mlquery$ itself does not expose or materialize this dependency. Provenance capture is precisely the capability $\mlquery'$ adds. What step 1 and the argument above establish is only that the \emph{values} $\mlquery'$ materializes as provenance ($\predInst$, $\gradInst$, $\grad$) are exactly the intermediate values $\mlquery$'s \texttt{Predict} and \texttt{Gradient} relations already compute internally en route to the final weights.
\item \emph{Completeness (no missing tuples)} Since $\modelprov$ carries exactly one live tuple per feature $C$ at each $j$ and the join with $\trainv$ is on $C$, every $t_i \in \trainRel$ with a matching feature $C$ contributes exactly one row to \texttt{recProv}; since every training instance has a feature value for every relevant $C$ by construction of the verticalized encoding (\Cref{sec:formal-sem}), $|\texttt{recProv}| = |\trainRel|$ at each $j$, so every $t_i \in \trainRel$ is represented matching \Cref{def:local-prov}'s requirement that $\localProv$ be defined for \emph{every} $t_i \in \trainRel$.
\item \emph{Value correctness} By steps 2--3 above, the materialized $\predInst$, $\gradInst$, and $\grad$ columns are exactly the values $f(t_i.v_c, p_c)$, $g(\predInst, t_i)$, and $\sum_{t_i \in c} \gradInst$ required by \Cref{def:local-prov}'s augmentation function $\mathcal{A}$.
\end{itemize}
Thus, $\localProv$ at step $j$, as materialized by $\mlquery'$, equals $\localProv(\trainRel, \modelRel)$ of \Cref{def:local-prov} exactly.

\mypara{Global provenance} By \Cref{def:global-prov}, $\globalProv = \bigcup_{j=0}^{k} \localProv$. Since (ii) shows $\mlquery'$ materializes $\localProv$ correctly and completely at \emph{every} $j = 0, \dots, k$ (via the sequential subquery chain re-executed fresh at each fixpoint step, as required by the ``Challenges of the Extension'' discussion in \Cref{sec:provqml}), and \Cref{sec:opt-global} establishes that the external accumulation of these per-iteration snapshots into RDDs is lossless (each $\localProv$ snapshot is retained, indexed by $J$, until subsumed by the accumulated union --- eviction occurs only \emph{after} subsumption), the accumulated result equals $\bigcup_{j=0}^{k} \localProv = \globalProv$ exactly, matching \Cref{def:global-prov}.

\mypara{Conclusion} By induction, (i) holds for all $j \leq k$, so $\mlquery'$ returns the same optimized model weights as $\mlquery$; by (ii) and the global provenance argument, $\mlquery'$ additionally computes $\localProv$ at every iteration and $\globalProv$ over all iterations exactly as defined. 
Thus $\mlquery'$ is both semantics-preserving and provenance-complete with respect to $\localProv$ and $\globalProv$.
\end{proof}
}

\section{Optimization}
\label{sec:optimization}

\begin{figure}[t]
    \centering \vspace{-3mm}
    \includegraphics[width=0.75\linewidth, trim= 0 60 0 0]{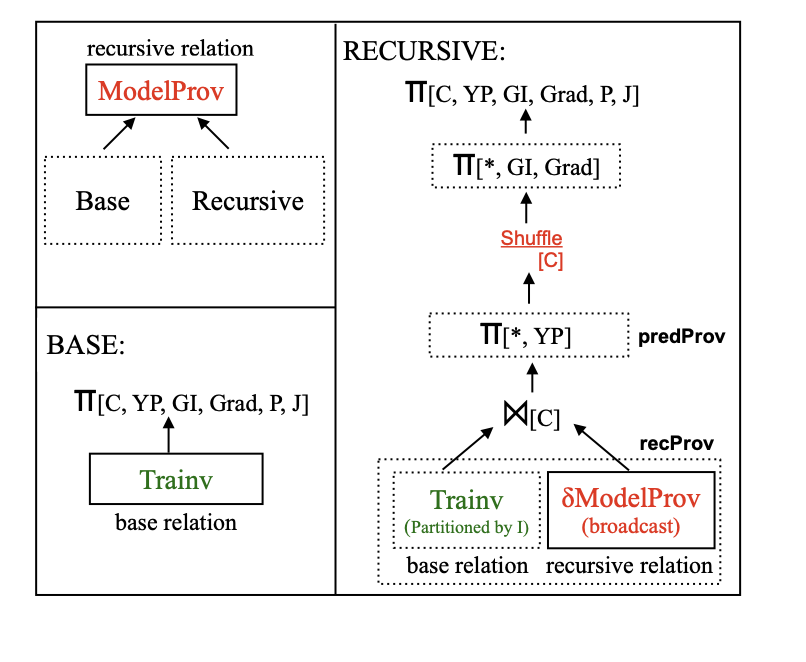}
     \vspace{-4mm}
    \caption{Optimized query plan for $\mlquery'$, 
    illustrating broadcast of $\delta\modelRel$, upfront 
    partitioning of $\trainRel$ by I, and 
    the single Shuffle on C per iteration.}
    \label{fig:opt-query-plan}
\end{figure}


Although the \syst's instrumentation algorithm (\Cref{sec:query-rewriting-semantics}) eliminates mutual and non-linear recursions, capturing $\localProv$ and $\globalProv$ during recursive evaluation introduces two distinct computational challenges. First, computing $\localProv$ at each iteration $j$ requires evaluating window functions over the relation \texttt{recProv}, which requires sorting and shuffling data that can dominate execution time. Second, maintaining $\globalProv$ across all iterations requires persisting $\localProv$ at every $j$, which causes intermediate result sizes to grow as $O(|\trainRel| \cdot k)$ where $k$ is the maximum number of iterations. 

\subsection{Optimizing Local Provenance Capture}
\label{sec:opt-local}

Evaluating $\mlquery'$ to capture $\localProv$ relies on window functions 
over the relation \texttt{recProv} to compute $\predInst$ and $\grad$ (e.g.,~\Cref{fig:recqmlquery-bgdlr}). In a distributed setting, window functions require data to be partitioned by a specific key before aggregation, which triggers a shuffle operation to redistribute tuples across worker nodes. 
Specifically, computing $\predInst$ requires partitioning by \lstinline{(J,I)} and computing $\grad$ requires partitioning by \lstinline{(J,C)}. If these shuffles are performed naively at every iteration $j$, the overhead accumulates linearly with the number of iterations, significantly degrading the performance.

We address this with two complementary techniques. First, since $\modelRel_j$ is much smaller than $\trainRel$ (it contains one tuple per feature dimension, while $\trainRel$ contains one tuple per feature per training instance), we broadcast $\modelRel_j$ to all worker nodes at each iteration. This allows \texttt{recProv} to be formed via a local join on each partition of $\trainRel$ without any shuffle of the training data. Second, we partition $\trainRel$ by the instance identifier \lstinline{I} once at the beginning of query evaluation and reuse this partition throughout training. Since \texttt{recProv} inherits the defined partition, it is already co-partitioned with the prediction window function (e.g., \lstinline{1/(1+exp(-(sum(V*P)) OVER (PARTITION BY} \lstinline{J,I)} in~\Cref{fig:recqmlquery-bgdlr}), eliminating its reshuffle. Note that since \lstinline{J} is uniform within each recursive step, the effective partitioning key is driven by \lstinline{I}, making the upfront partition on $\trainRel$ by \lstinline{I} sufficient. The gradient window function (e.g., \lstinline{sum((YP-Y)*V) OVER (PARTITION BY} \lstinline{J,C)} in the figure) requires a single reshuffle of \texttt{recProv} from \lstinline{I} to \lstinline{C} at each iteration $j$, which is unavoidable. Partitioning $\trainRel$ by \lstinline{C} instead would introduce two reshuffles per iteration, one for \texttt{predictProv} from \lstinline{C} to \lstinline{I}, and another back to \lstinline{C} for the gradient computation. Thus, partitioning by \lstinline{I} minimizes total shuffle overhead to exactly one reshuffle per iteration. 
\Cref{fig:opt-query-plan} shows the resulting optimized query plan. $\trainRel$ (\textcolor{darkgreen}{$\trainv$} partitioned by \lstinline{I}) and $\delta\modelRel$ (\textcolor{red}{$\modelprov$} broadcasted) form \texttt{recProv} via a local join with no shuffle. The single \textcolor{red}{\underline{Shuffle} [C]} between \texttt{predictProv} and \texttt{Grad} (e.g., lines 14--17 in~\Cref{fig:recqmlquery-bgdlr}) is the only inter-node in the plan redistributing data required per iteration.
Furthermore, since the size and structure of \texttt{recProv} for evaluation at each iteration remain fixed at $|\trainRel|$ tuples across all iterations (\Cref{sec:opt-global}), the sorted order within each partition produced by the gradient reshuffle can be cached and reused at every subsequent $j$, reducing the sort cost within a partition.

\enlargethispage{0.8\baselineskip}
This optimization is valid for the ML algorithms in \Cref{tab:queryconfig} because the join key between $\trainRel$ and $\modelRel_j$ is the feature dimension \lstinline{C}, 
and the aggregation keys for $\predInst$ and $\grad$ are \lstinline{I} and \lstinline{C}, regardless of the specific predict function $f$ or gradient function $g$. Furthermore, since $\trainRel$ is read-only throughout training, its partitions remain valid across all iterations, making upfront partitioning a sound and stable optimization. This mirrors the replica mechanism of~\cite{WW21}, which similarly avoids repeated shuffles of the base relation by co-partitioning it with the recursive relation on the required join keys.
Note that this optimization extends to any ML algorithm whose training loop joins a fixed base relation with a small delta model relation on a stable key and computes per-instance and per-feature aggregations on fixed partition keys. 
Algorithms with additional aggregation levels or multiple join keys would require a corresponding extension of the partitioning strategy.

\subsection{Optimizing Global Provenance Capture}
\label{sec:opt-global}
$\globalProv$ is defined as the union of $\localProv$ across all $k$ fixpoint iterations (\Cref{def:global-prov}). Since $\localProv$ at each iteration $j$ has size $O(|\trainRel|)$, maintaining $\globalProv$ naively requires $O(|\trainRel| \cdot k)$ space. For large training datasets and number of iterations, the intermediate result can far exceed the size of the final model output, placing significant memory and I/O overhead. A further overhead arises from the join in \texttt{recProv}. Without careful management, the join between $\trainRel$ and $\modelRel$ at each iteration $j$ could include model states from prior iterations that are still resident in memory, producing stale or redundant provenance that inflates $\globalProv$ unnecessarily.

To address this issue, we adopt a strategy analogous to the delta-based optimization of~\cite{WW21} through DSN (\Cref{sec:dsn}). At each iteration $j$, only the delta relation $\delta\modelRel_j$, i.e., the model parameters updated in the most recent fixpoint step, is used as the input to the join in \texttt{recProv}.
Consequently, \texttt{recProv} at each $j$ joins $\trainRel$ exclusively with $\modelRel_{j-1}$, producing exactly $|\trainRel|$ tuples with no contribution from the previous model states. This is correct because the join key between $\trainRel$ and $\modelRel$ includes the iteration counter \lstinline{J}, so tuples from earlier iterations carry mismatched \lstinline{J} values and would produce no valid join results, analogous to the observation in~\cite{WW21}. 
To maintain $\globalProv$, the $\localProv$ captured at each iteration $j$ is accumulated as a chain of RDDs, preserving the full optimization trajectory across all fixpoint iterations. To bound memory usage, the $\localProv$ RDD from iteration $j-1$ is unpersisted at the end of each fixpoint step once its content is subsumed by the accumulated $\globalProv$, while remaining accessible via Spark's DAG lineage on demand. 
At any point during the recursion evaluation, only the cached $\trainRel$ and the $\localProv$ snapshot from the latest iteration $j$ are physically resident in memory which keeps the memory usage bounded at $O(|\trainRel| + |\localProv|) = O(2 \cdot |\trainRel|)$ per iteration. 
Including older model states would produce no additional tuples, while retaining them in memory would incur unnecessary overhead. The incremental accumulation of $\globalProv$ as a union of per-iteration $\localProv$ snapshots is similarly justified. 
Since $\globalProv$ can be queried post-hoc for explanation or debugging 
rather than consumed during training, evicting $j-1$ RDDs after each fixpoint step introduces no correctness concern, as the accumulated $\globalProv$ already subsumes all prior $\localProv$ snapshots by \Cref{def:global-prov}.

\section{Experiments}
\label{sec:experiments}

We start by describing our experimental setup in~\Cref{sec:setup}\ifnottechreport{.} \iftechreport{followed by an end-to-end performance evaluation in~\Cref{sec:end-to-end}.}
We then investigate the impact of several performance factors in~\Cref{sec:perf}, and measure provenance capture overhead by comparing with Radlog~\cite{WW21} that evaluates the ML query $\mlquery$ in~\Cref{sec:prov-overhead}. Finally, in~\Cref{sec:discussion}, we provide a cost estimation for computing explanations like feature importance.


\subsection{Experimental Setup}
\label{sec:setup}

\mypara{Environment}
All queries were executed on a cluster with 7 worker nodes. Each node 
has 2 x 3.3Ghz Intel Xeon Gold 6234 CPUs
(16 cores)
and 192 GB memory.
The cluster hosts 35 executors of 32GB main memory and 4 CPU cores each.

\mypara{Datasets} We conduct experiments on four real-world datasets: Income ~\cite{ACSincomedataset21}, BikeSharing ~\cite{bikesharingdataset13}, SpamBase ~\cite{spambasedataset99} and WebSpam ~\cite{webspamdataset12}. 
Income is a widely used benchmark dataset with binary labels. 
BikeSharing records the number of bikes rented over a given time period based on weather and seasonal information, SpamBase labels emails as spam or non-spam, while
WebSpam, a LIBSVM ~\cite{libsvm} dataset, distinguishes between legitimate web pages and spam.
The datasets used in this experiment are summarized in ~\Cref{tab:Experiments Data} where \#Tuple means the number of training instances, while \#Feature means the number of dimensions in the dataset. 
Each data is verticalized in the form of $\langle I,Y,l,v \rangle$ explained in~\Cref{sec:formal-sem}.

\mypara{Queries} 
The ML query sets include linear, lasso, ridge, and logistic regression models plus SVM denoted as queries $\query_1$ through $\query_5$ in~\Cref{tab:queryconfig}.
The table provides an overview of the query configurations used for our experiments.
Across all five queries, 
the ML query structure differs in the predict $f$ and gradient $g$ functions.
For experiments with the BikeSharing dataset we use queries $\query_1$, $\query_2$, and $\query_3$ that express regression tasks and for Income, SpamBase and WebSpam, we use $\query_4$ and $\query_5$ that encode classification tasks.

\mypara{Parameters and Metrics}
In this experiment, we focus on the query execution time as the correctness of the optimized weights was validated by comparing them with standard implementations in scikit-learn.
Preprocessing time for data verticalization is excluded. 
Learning rate, \lstinline{lr}, is set to $0.01$, \lstinline{j} is set to 100 maximally, and \lstinline{n} depends on the number of tuples in each dataset (\Cref{tab:Experiments Data}).
%

\begin{table}
\caption{Experiment Dataset Statistics}
\vspace{-4mm}
\setlength{\tabcolsep}{8pt} 
\renewcommand{\arraystretch}{0.9}
\label{tab:Experiments Data}
\begin{tabular}{lll|lll}
\toprule
\multicolumn{3}{r|}{A. Verticalized Dataset} & \multicolumn{2}{c}{B. Original Dataset} \\
\cmidrule(r){1-3} \cmidrule(l){4-5}
Dataset & \#Tuple & \#Feature & \#Tuple & \#Feature \\ 
\midrule
Income & 17M & 10 & 1.7M & 10  \\
BikeSharing & 226K & 13 & 17K & 13 \\ 
SpamBase & 262K & 57 & 4.6K & 57  \\
WebSpam & 29M & 138 & 350K & 138 \\
\bottomrule
\end{tabular}
\end{table}

\begin{figure*}[t]
    \centering \vspace{-2mm}
    \begin{minipage}[t]{0.49\textwidth}  
        \vspace{0pt} 
        \centering
        \includegraphics[width=\linewidth,trim= 0 5 0 0,clip]{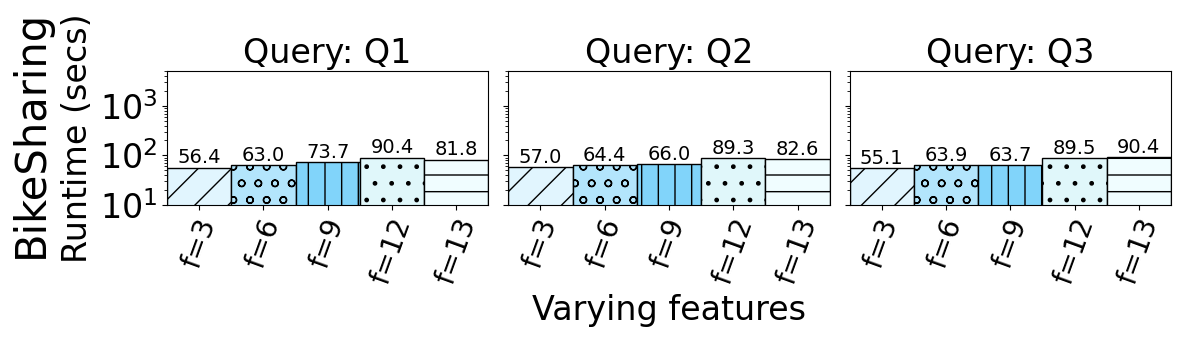}
        \includegraphics[width=\linewidth,trim= 0 6 0 0,clip]{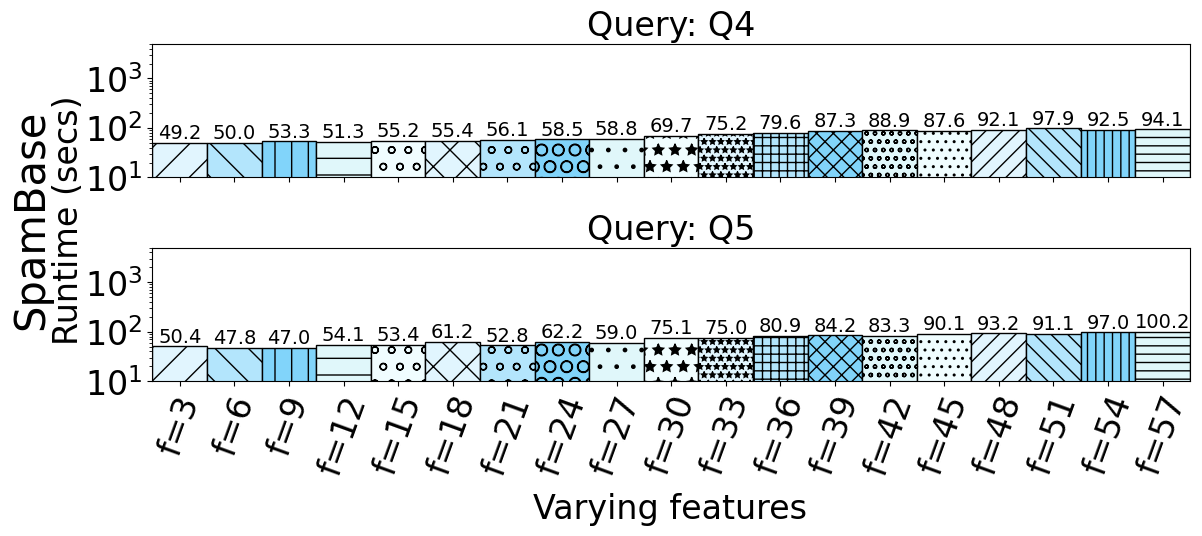}
        \includegraphics[width=\linewidth,trim= 0 5 0 0,clip]{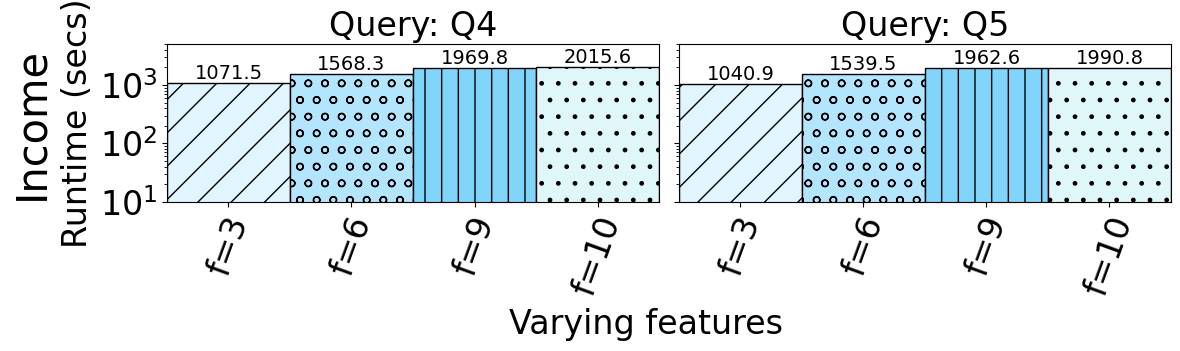}
        \vspace{-8mm}
        \caption{Runtime for varying number of features 
        } 
        \label{fig:varying_features}
    \end{minipage}
    \hfill
    \begin{minipage}[t]{0.5\textwidth}
        \vspace{0pt} 
        \centering
        \includegraphics[width=\linewidth,trim= 0 5 0 0]{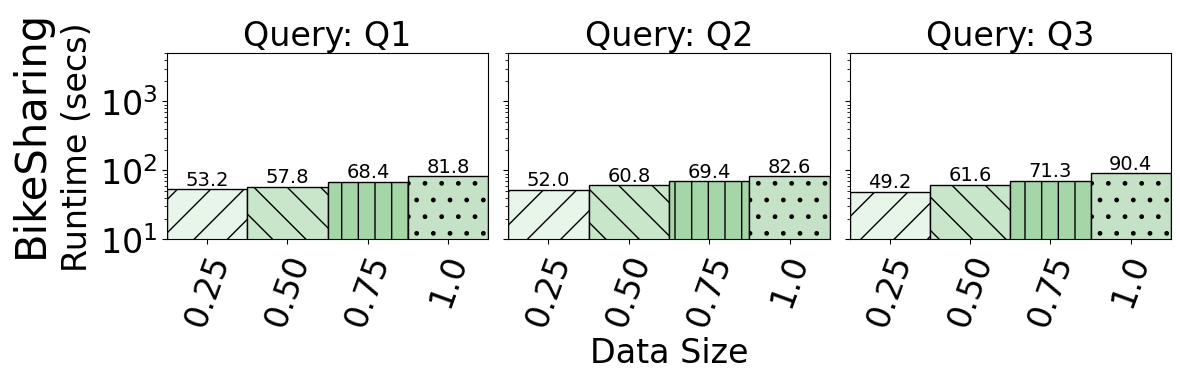}
        \includegraphics[width=\linewidth,trim= 0 5 0 0]{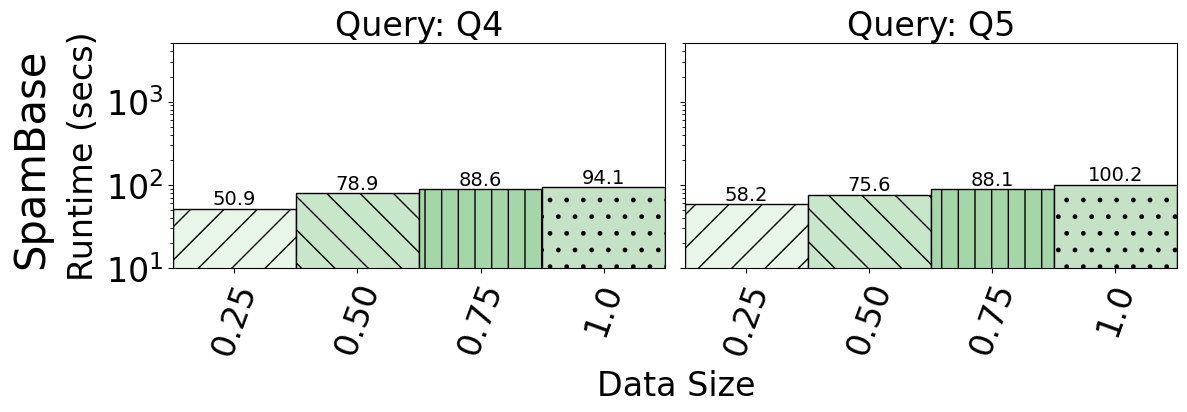}
        \includegraphics[width=\linewidth,trim= 0 5 0 0,]{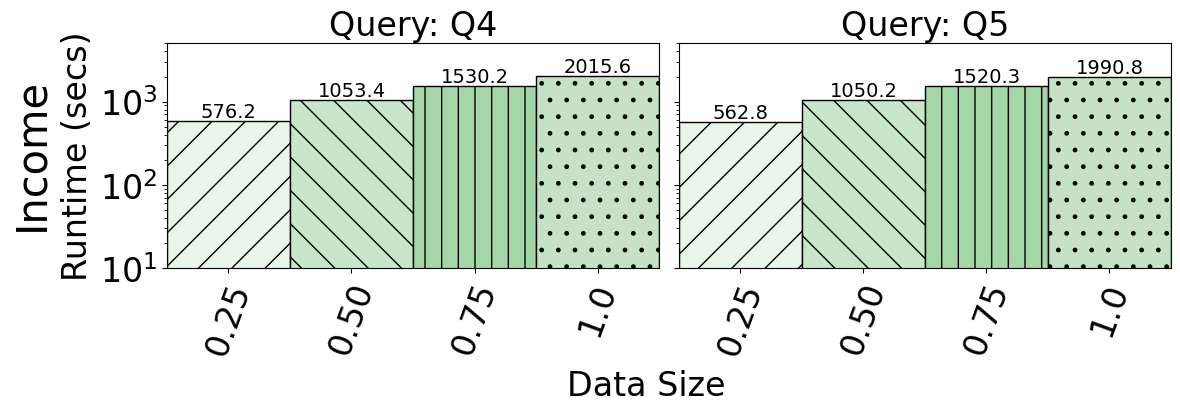}
        \vspace{-8mm}
        \caption{Runtime for varying input data sizes
        } 
        \label{fig:varying_datasize}
    \end{minipage}
\vspace{-2mm}
\end{figure*}

\iftechreport{
\begin{figure}[t]
\centerline{\includegraphics[width=0.75\columnwidth, trim= 0 30 0 20]{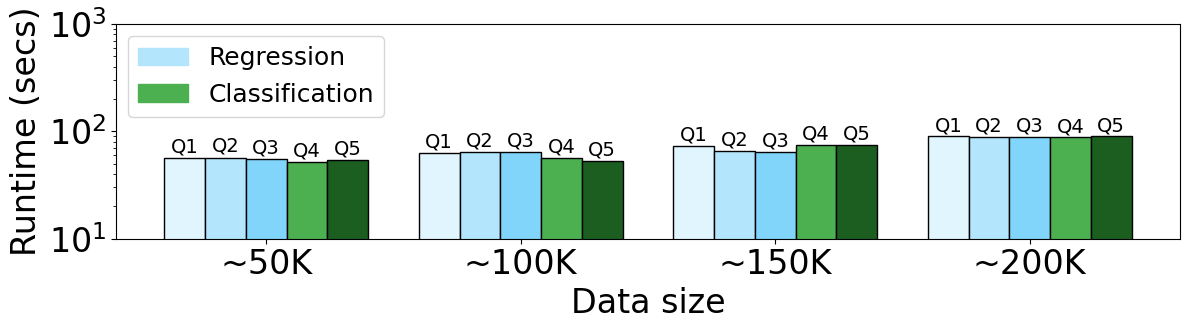}}
\Description{bar chart showing performance comparison, training across model complexities}
\vspace{-2mm}
\caption{End-to-end execution time for $\syst$ queries with regression workload ($\query_1$-$\query_3$) on BikeSharing and classification workload ($\query_4$-$\query_5$) on SpamBase}
\label{fig:modelcomplexities}
\end{figure}

\subsection{End-to-End Performance}
\label{sec:end-to-end}
We evaluate the performance of the queries 
over dataset sizes ranging from $\sim$50K to $\sim$200K tuples.
\Cref{fig:modelcomplexities} shows end-to-end execution time across all queries. Notably, runtime increases consistently sublinearly with dataset size. Regression queries ($\query_1$ - $\query_3$) show closely similar performance, which indicates that regularization has minimal impact. A similar pattern holds for classification queries ($\query_4$ - $\query_5$). 
Overall, the results show that 
\syst's provenance capture queries for various ML models do not introduce disproportionate overhead as data scales. 
}

\subsection{Performance Impacting Factors}
\label{sec:perf}
\iftechreport{
To better understand
the behavior observed in Section 7.2, we investigate}
\ifnottechreport{We first investigate}
the runtime of each $\syst$ query w.r.t. feature dimensionality, dataset size, the number of iterations \lstinline{j}, 
\ifnottechreport{\textcolor{MR}{and the number of worker nodes.}}
\iftechreport{the number of worker nodes, and the impact of shuffle optimization.}
Note that while provenance capture over smaller datasets such as BikeSharing and SpamBase fit within a single machine’s capacity, the single-machine experiments on Income (17M tuples) and WebSpam (29M tuples) did not complete due to out-of-memory, confirming that distributed execution is necessary at this scale. The following experiments are therefore conducted in the
distributed environment described in~\Cref{sec:setup}.

\setlength{\tabcolsep}{6pt} 
\renewcommand{\arraystretch}{0.9} 
\begin{table}[t]
\centering
\caption{Feature-Subset Data Statistics}
\vspace{-4mm}
\begin{tabular}{|c|c|c|c|c|c|}
\hline
\multicolumn{2}{|c|}{\textbf{BikeSharing}} & \multicolumn{2}{c|}{\textbf{SpamBase}} & 
\multicolumn{2}{c|}{\textbf{Income}} \\
\hline
\#Feature & \#Tuple & \#Feature & \#Tuple  & \#Feature & \#Tuple \\
\hline
3 & $\sim$52K  & 12 & $\sim$55K & 3 & $\sim$5M \\
6 & $\sim$104K & 21 & $\sim$97K & 6 & $\sim$10M \\
9 & $\sim$156K & 33 & $\sim$152K & 9 & $\sim$15M \\
12 & $\sim$209K & 45 & $\sim$207K & 10 & $\sim$17M \\
13 & $\sim$226K & 51 & $\sim$235K & - & - \\
- & - & 57 & $\sim$262K & - & - \\
\hline
\end{tabular}
\label{tab:varying-features-dataset-sizes}
\end{table}
%
\mypara{Varying Number of Features}
We evaluate the impact of varying number of features on $\syst$ query runtime. 
The number of features $|f|$ is varied for the BikeSharing dataset $3 \le |f| \le 13$, SpamBase $3 \le |f| \le 57$, and Income $3 \le |f| \le 10$, while j = 100.
The corresponding dataset for each $f$ is created by selecting the corresponding subset of the verticalized training data.
\Cref{tab:varying-features-dataset-sizes} summarizes the data size for each feature subset across the datasets. 
\Cref{fig:varying_features} shows the result. The runtime increases as the number of features increases, and this holds across all queries. This is because higher dimensionality increases the size of the verticalized data, and thus, the size of provenance. 
Consequently, the cost of computing per-instance predictions and gradients grows.
For smaller feature sets, runtime remains relatively stable within each dataset, as lower dimensionality reduces shuffle and node communication overhead.
In general, the runtime over Income data is higher than the one over BikeSharing.
For Income, runtime at $f=9$, is nearly twice that at $f=3$, whereas for BikeSharing, the increase from $f=3$ to $f=9$ is only a few seconds. This difference reflects the varying data scale per feature, as Income at $f=3$ has 5M tuples and $f=9$ has 15M tuples, compared to 52K and 156K tuples for BikeSharing.

\mypara{Varying Dataset Size}
Next, we vary input data size  
from 25\% (0.25) to 100\% (1.0) of the complete dataset, proportionally sampling training instances across all features to preserve the feature distribution at each scale. 
The number of iterations is fixed at 100.
\Cref{fig:varying_datasize} shows the results.
Runtime increases with data size across all queries and datasets since larger inputs increase both the shuffle volume during recursive evaluation and the amount of provenance tracked per iteration.
For the smaller datasets such as BikeSharing and SpamBase, runtime grows sub-linearly while, for the larger dataset like Income, runtime scales roughly linearly. 
This behavior reflects the increased cost of scanning larger relations and aggregating more tuples per iteration. 
Overall, the results demonstrate that $\syst$ efficiently leverages data partitioning and parallelism to keep provenance capture tractable at scale.

\mypara{Varying Number of Iterations}
We vary the number of iterations \lstinline{j} from 5 to 100 and measure runtime for the queries $\query_1$, $\query_3$, and $\query_4$. 
\Cref{fig:iterations-across-all-queries} shows the results. Runtime increases with \lstinline{j} across all queries and datasets since each iteration adds an additional cost of recursive computation and provenance tracking.
The runtime is sub-linear, which reflects the fixed cost per-iteration of joining $\trainRel$ with $\delta\modelRel$ and the amortized effect of upfront partitioning. 

\mypara{Impact of Large Number of Features}
We conducted a stress test for $\syst$ using $\query_4$ and WebSpam dataset which has the largest number of features.
The runtime measured by varying the number of iterations and dataset size is shown in~\Cref{fig:iterations-and_size_webspam_Q4}.
\syst\ successfully executes the query with j = 100 over 7M, but the maximum number of iterations that can be completed decreases as dataset size grows, reflecting the increased memory pressure at larger scales.
This is also visible in the result of varying data size. 
Only runtime over 25\% data reaches j = 100, while larger dataset sizes do not complete the full 100 iterations.

\noindent \textcolor{MR}{
\mypara{Impact of Number of Workers}
This experiment evaluates the impact of varying the number of worker nodes from 2 to 7 using $\query_4$ on the Income dataset. \Cref{fig:varying-worker-nodes} reports the measured runtimes for varying numbers of features (left) and varying data sizes (right).
The result on the left confirms that fewer worker nodes (e.g., 2$\sim$3) can process smaller workloads ($\sim$6 features), but fail to complete larger workloads ($\sim$10 features) due to insufficient memory capacity. A similar trend is observed in the result on the right. Smaller data sizes ($\sim$8M tuples equivalent to 50\% of the dataset) are successfully processed with fewer nodes, while larger data sizes (over 75\% of the dataset) require more worker nodes to complete successfully.
}

\noindent
\textcolor{MR}{
\mypara{Impact of Shuffle Optimization}
This experiment evaluates the impact of reduced shuffle optimization of $\syst$ using $\query_1$ 
\iftechreport{and $\query_4$} on the BikeSharing 
\ifnottechreport{dataset.}
\iftechreport{and SpamBase datasets respectively.} \Cref{fig:with-and-without-shuffle-bike} shows the measured runtime with and without the shuffle optimization, for varying numbers of features (left) and data sizes (right).
Overall, the result shows that execution without the shuffle optimization incurs higher runtime, with an overhead of up to $3.4\times$ compared to the optimized execution.
}

%
%
\begin{center}
\begin{figure}[t]
    \vspace{-3mm}
    \centering
    \subfigure{%
        \includegraphics[width=0.75\textwidth, trim= 0 30 0 0]{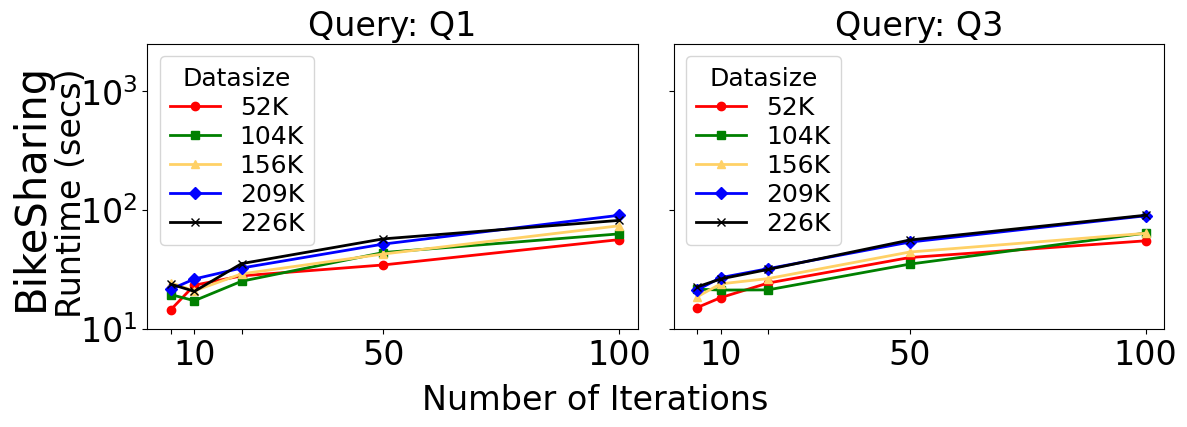}
        \label{fig:itrbike}
    }
    \subfigure{%
        \includegraphics[width=0.75\textwidth, trim= 0 30 0 0]{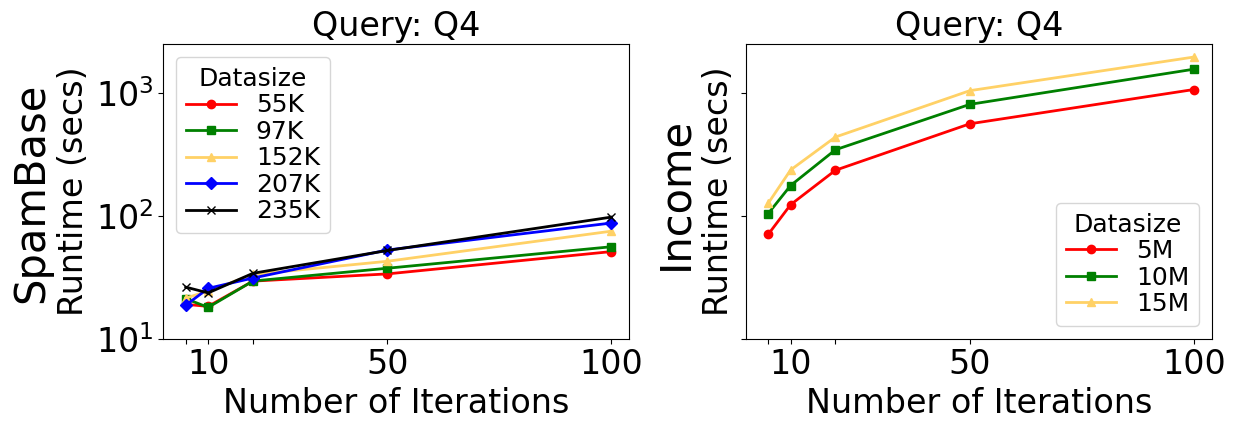}
        \label{fig:itrbike}
    }
    \vspace{-2mm}
    \caption{Runtime varying the number of iterations}
    \label{fig:iterations-across-all-queries}
\end{figure}
\end{center}
\begin{figure}[t]
\vspace{-2mm}
{\includegraphics[width=0.75\linewidth, trim= 0 20 0 0]{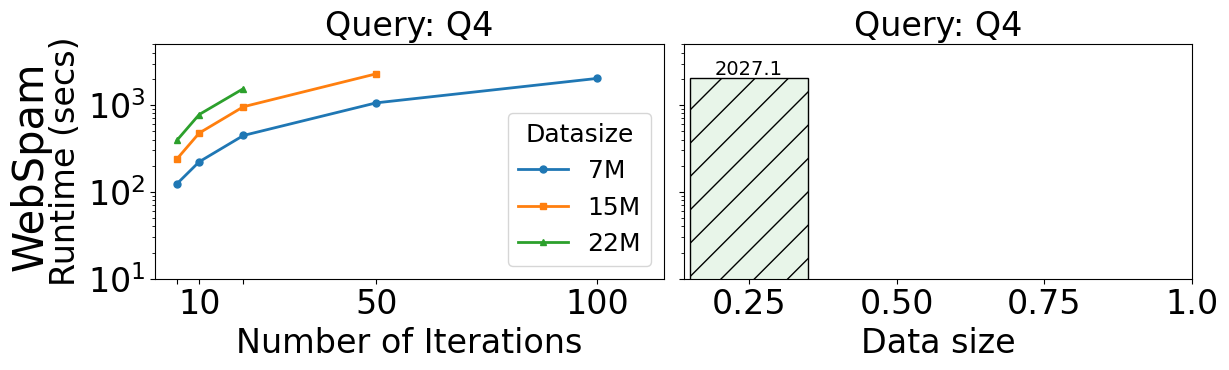}}
\vspace{-2mm}
\caption{Runtime for large number of features
}
\label{fig:iterations-and_size_webspam_Q4}
\end{figure}
%
\begin{figure}[t]
\vspace{-2mm}
{\includegraphics[width=0.75\linewidth, trim= 0 25 0 0]{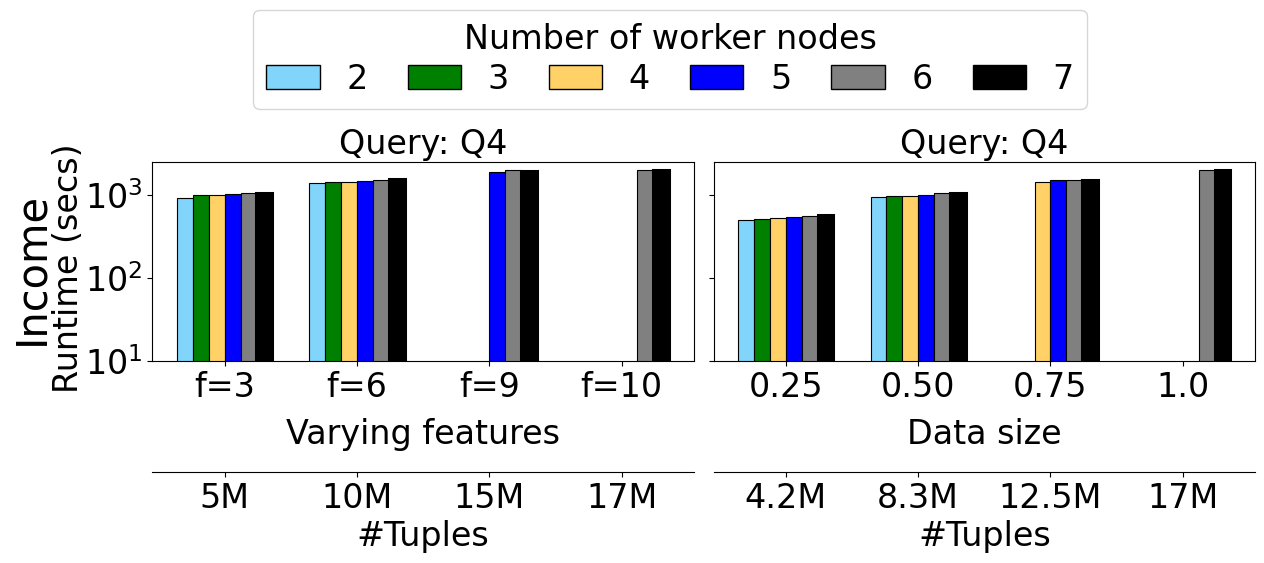}}
\vspace{-1mm}
\caption{Runtime for varying number of worker nodes}
\label{fig:varying-worker-nodes}
\end{figure}
%
%
\begin{figure}[t]
\ifnottechreport{
\centerline{\includegraphics[width=0.75\columnwidth, trim= 0 0 0 0]{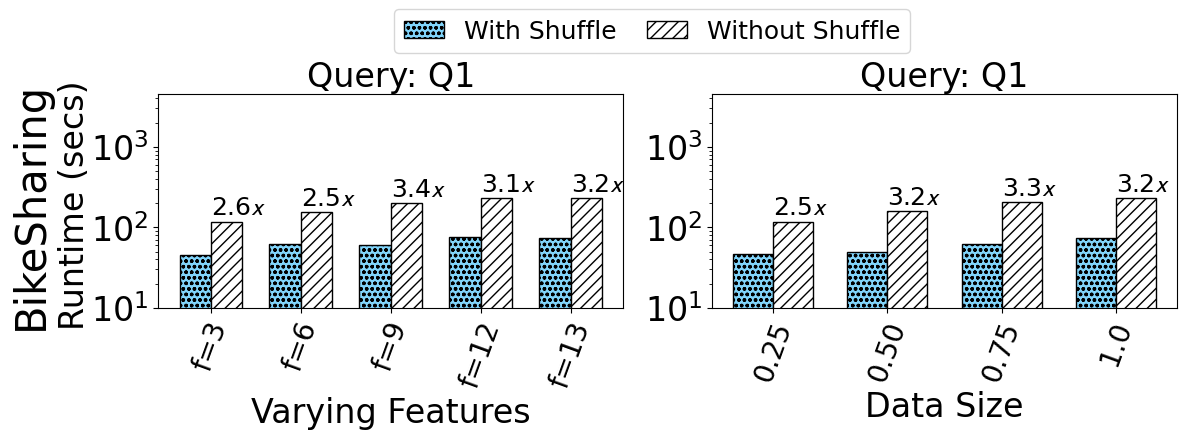}}}
\iftechreport{
\centerline{\includegraphics[width=0.75\columnwidth, trim= 0 0 0 0]{img/with-and-without-shuffle-bike.png}}
\centerline{\includegraphics[width=0.75\columnwidth, trim= 0 0 0 0]{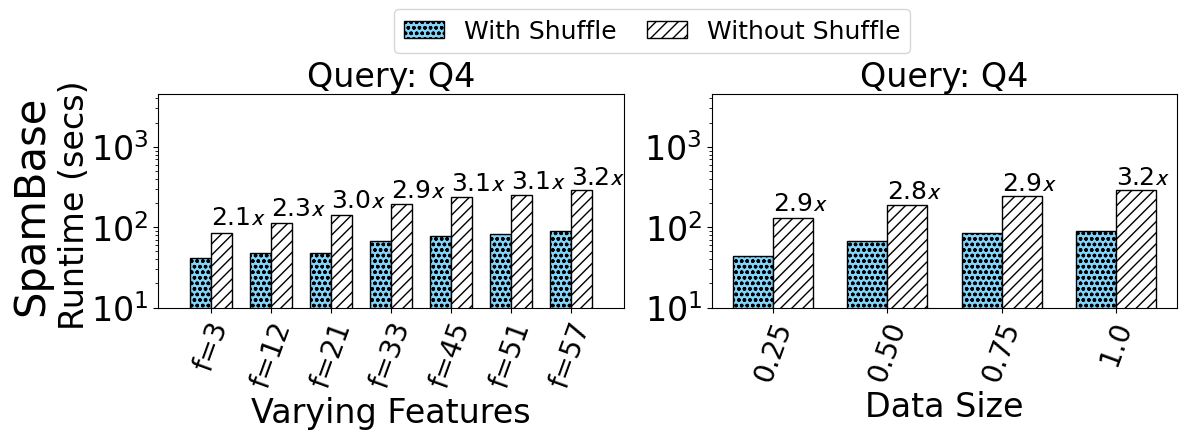}}}


\vspace{-4mm}
\caption{\textcolor{R1}{Runtime for impact of shuffle optimization}}
\label{fig:with-and-without-shuffle-bike}
\end{figure}


\begin{figure}[t]
\vspace{-2mm}

{\includegraphics[width=0.75\linewidth, trim= 0 0 0 0] {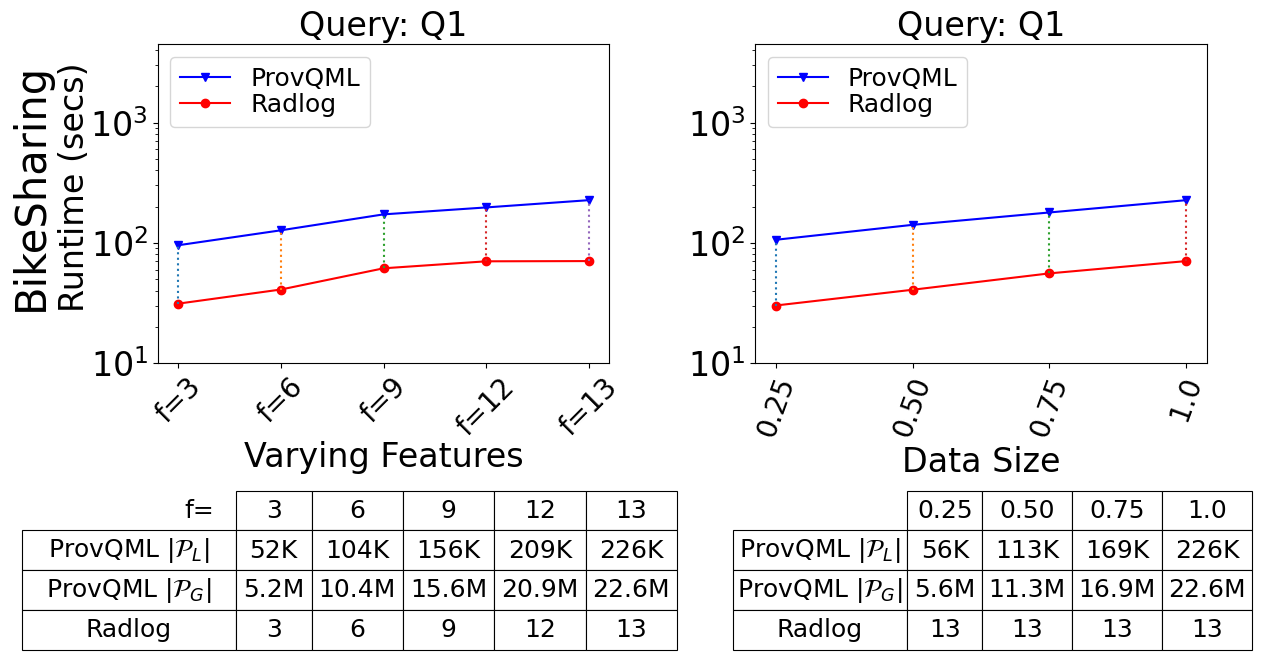}} 
{\includegraphics[width=0.75\linewidth, trim= 0 0 0 0] 
{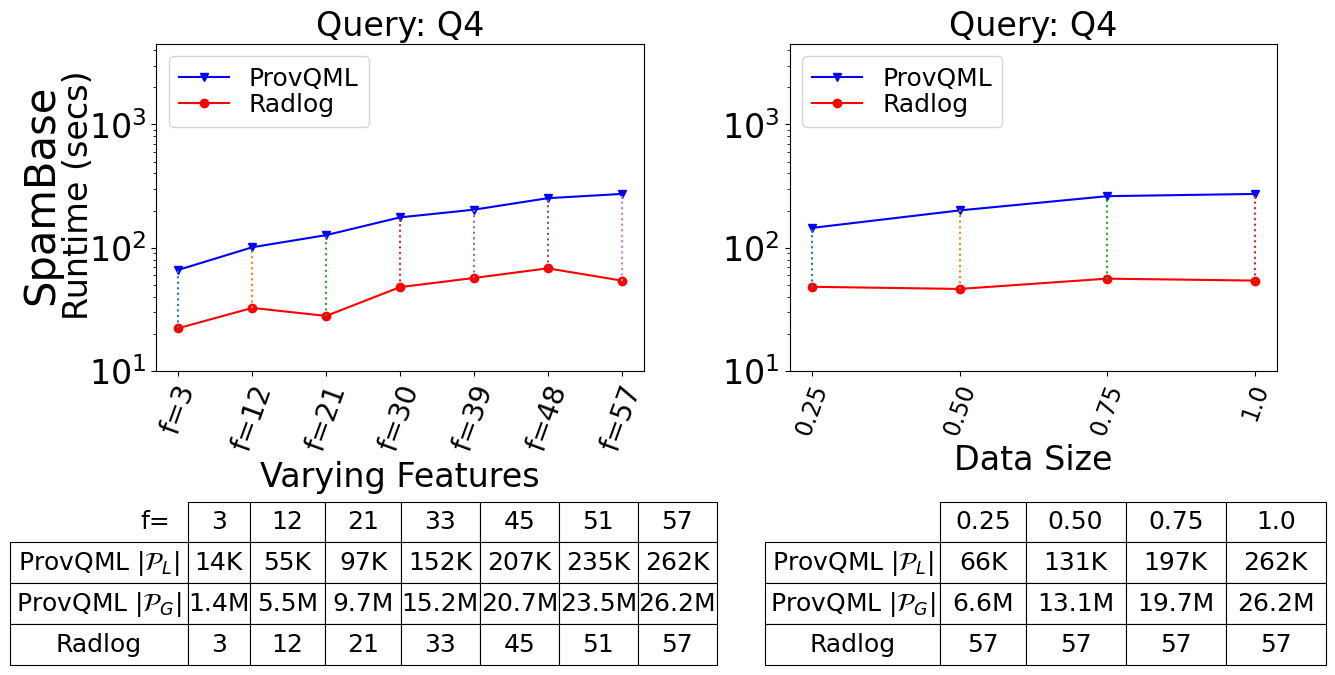}}
\vspace{-3mm}


\caption{\textcolor{R3}{Overhead of provenance for $\query_1$ and $\query_4$ across varying number of features and dataset sizes} 
}
\label{fig:provenance-cost-provqml-radlog}
\end{figure}

\vspace{-2em}
\subsection{Evaluating Provenance Overhead}
\label{sec:prov-overhead}
We measure the provenance computation overhead by comparing the runtime of $\syst$ with the runtime of Radlog ~\cite{WW21}.
Radlog is a recursive aggregate Datalog query engine for ML algorithms built on top of Apache Spark and computes the optimized model's parameters (executing the Datalog query equivalent to the SQL in~\Cref{fig:formalsemanticsquery}). 
As demonstrated in~\cite{WW21}, the system achieves the best end-to-end execution time for the ML models in~\Cref{tab:queryconfig} compared to MLlib~\cite{mllib16} and SystemML~\cite{systemML11}.
Thus, we only compare with it.
As Radlog does not capture provenance for $\mlquery$, 
we take Radlog as a baseline to measure the overhead of capturing provenance.
Due to issues in configuring Radlog in our cluster (\Cref{sec:setup}) (lack of resources to resolve them), where all other experiments were conducted, this evaluation was instead performed on a machine with 2.9 GHz quad-core Intel CPU and 16 GB memory.
\Cref{fig:provenance-cost-provqml-radlog} shows the runtime of $\query_1$ (linear regression) on the BikeSharing dataset and $\query_4$ (logistic regression) on the SpamBase dataset for both varying the number of features, $3 \le |f| \le 13$ for $\query_1$ and $3 \le |f| \le 57$ for $\query_4$,
and dataset sizes (25\% - 100\%).
Provenance introduces a consistent considerable overhead. 
For $\query_4$ varying $|f|$, the provenance capture overhead is at most $5\times$ slower than evaluating the model weights alone, occurring at the full dataset with $|f| = 57$ where \syst\ took 272.80 seconds compared to 54 seconds for Radlog. A similar trend is observed when varying the dataset size, where the provenance overhead reaches
$4\times$. The same trend holds for $\query_1$.
\textcolor{MR}{The provenance capture overhead observed aligns with overheads reported in the data provenance literature for comparably expressive query classes ~\cite{DM22, MDG18}.} 
The tables shown in \Cref{fig:provenance-cost-provqml-radlog} present the result size of Radlog and the size of provenance captured by $\syst$.
The size of local provenance $\localProv$ is bounded by $|\trainRel|$, because each output tuple in $\modelRel_j$ (the model parameter per feature dimension at $j$th iteration) is mapped to each tuple in $\trainRel$ (a tuple per feature per training instance). 
Thus, the size of global provenance $|\globalProv| = |\localProv| * k$ across $k$ iterations.
In contrast, the output size of $\mlquery$ is $|\modelRel_k|$, containing the optimized model parameters.
For \syst\, the reported size corresponds to 
$|\localProv|$ for a single iteration, while the overall provenance size $|\globalProv| = |\localProv|$ * 100 as the evaluation is conducted with 100 iterations. For example, the provenance size for $\query_1$ on the BikeSharing dataset for $|f|=3$ is 5.2M = 52K $\cdot$ 100 tuples.


\subsection{Discussion}
\label{sec:discussion}
One of the widely used methods for computing model explanation is marginal feature contribution, formalized by Shapley value (SV)~\cite{Shapley}. SV attributes the contribution of each feature to a model output as the weighted average of its marginal contributions across all possible feature coalitions.
The computation introduces substantial overhead as SV requires an exhaustive evaluation of all possible feature coalitions, where each coalition is a model. 
%
Despite this expensive cost, we observe the feasibility of computing such explanations by extending \syst\ based on the prior work~\cite{TA23,AL25}.
To test, we encoded all possible coalitions of features as data, which we refer to as \emph{explanation data}.
For that, we augmented the input dataset with an additional attribute 
which encodes feature coalitions. 
This schema design provides two benefits: (i) it provides support for evaluations over all feature coalitions without requiring changes to the ML query and (ii) it shifts complexity that would otherwise arise at the query level to the data level. 
We implemented this extension in \syst\ by augmenting the recursive relation schema with an additional coalition identifier column, e.g., extending the \modelprov\ relation in Figure~\ref{fig:recqmlquery-bgdlr} to \lstinline{ModelProv(C, CF, YP, GI, Grad, last() AS P, max() AS J)}, where \texttt{C} and \texttt{CF} partition the training data by feature coalition and allow the existing query structure to evaluate all coalitions in a single recursive execution.

Interestingly, our investigation extracted that the computational overhead may not be completely unmanageable.
We run our experiments on the BikeSharing and SpamBase datasets using queries 
$\query_1$ and $\query_4$, varying the number of features. 
Note that, with the naive extension, the evaluation could be conducted for the small number of features.
The result is shown in~\Cref{fig:graph-and-table-cf-cost-spambase}.
The feature size and the corresponding data size are shown on the x-axis.
The experimental results highlight several observations. 
First, runtime scales roughly linearly with the number of feature coalitions. 
Secondly, the per-iteration cost increases more drastically as the intermediate result size is larger than the one for a single model.
Note that the runtime of this evaluation includes provenance capture time for every feature coalition.
A promising solution is relying on sampling that generates a representative subset of explanation data to compute explanations similar to the existing approximation-based methods in ML explanations~\cite{SR25}.
Without approximation, we envision that the computational overhead may be further reduced by leveraging the captured provenance across coalitions to avoid redundant recomputation, and by pruning coalitions that are unlikely to contribute significantly based on properties of the model or gradient structure.
This broadly captured information will enable new types of explanations, e.g., quantifying the influence of a group of training instances and a subset of features on the model's parameters.
Further investigation of such explanations is an interesting avenue for future work.

\begin{figure}[t]
\centering \vspace{-3mm}
{\includegraphics[width=0.75\columnwidth, trim= 0 45 0 0]{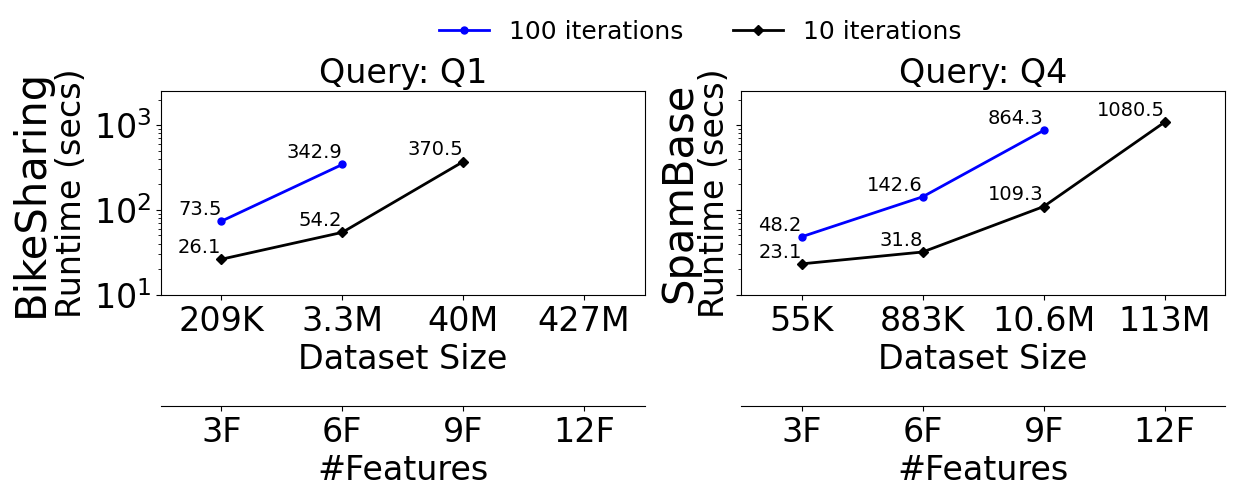}}
\vspace{1mm}

\caption{Runtime for evaluating all possible models
}
\label{fig:graph-and-table-cf-cost-spambase}
\end{figure}

\section{Related Work}
\label{sec:relatedwork}
\mypara{Recursive Query Processing}
Processing recursive queries in SQL~\cite{sql4ml19,indatabase21} and Datalog~\cite{SY16,WW21,mlog17} has been extensively explored. Support for aggregates within recursive queries was initially formalized in~\cite{KS1991}. 
Subsequent work showed that aggregates in recursion~\cite{SY16,SBKK2017,sql4ml19,WW21} increase the expressive power of recursive queries and provide formal declarative semantics for correctness, enabling ML workloads to be expressed in familiar declarative languages. 
However, none of these systems support provenance, which provides the information needed to explain ML model behavior.

\mypara{Database Provenance}
Provenance for database queries has received significant research attention over the past decade~\cite{BK01,green2007provenance,huang2008provenance,glavic2009perm,karvounarakis2010queryingDP,IS15,PUG18j,LL20}, including provenance for aggregate queries~\cite{AD11,AL23}. Provenance for recursive queries is understudied~\cite{DM22}.
Outside the database context, \cite{padovani2025ProvTracking} proposes a JSON-based library to capture provenance in ML systems, \cite{SBKK2017} tracks provenance for training runs and predictions, and~\cite{mlflow2prov23} captures provenance for ML tasks integrated with Git. 
However, these systems operate outside the database setting and require additional tools. In contrast, $\syst$ is a SQL-based system that leverages query rewriting to capture fine-grained provenance for ML workloads expressed as recursive aggregate SQL queries, without requiring additional integration.

\mypara{In-Database Machine Learning}
In-database machine learning integrates ML pipelines directly within database systems enabling tasks such as data preprocessing, model training, and validation. 
Systems such as MADlib~\cite{madlib12}, BigQuery ML~\cite{bigqueryml}, and Umbra~\cite{indatabase21} support in-database ML via SQL interfaces, but implement ML algorithms as user-defined functions or system-internal extensions rather than as recursive aggregate SQL queries. 
As a result, they neither expose the recursive structure of the ML computation nor support provenance capture during training.

\mypara{Distributed Machine Learning Systems}
Apache Spark ~\cite{apachespark12} has gained significant attention as a general platform for large-scale analytics and machine learning. 
SystemML~\cite{systemML11} supports ML workloads on Spark using R or Python, while BigDatalog~\cite{SY16} and Radlog~\cite{WW21} support declarative large-scale iterative computations on Spark using Datalog. 
\section{Conclusions}
\label{sec:conclusion-future-work}
In this paper, we presented $\syst$, a provenance-aware recursive aggregate SQL engine built on Apache Spark that efficiently captures fine-grained provenance while evaluating ML algorithms as recursive aggregate queries. Our approach instruments complex ML queries into queries with standard linear recursion, using window functions and a chain of provenance subqueries. 
We experimentally demonstrate that $\syst$ efficiently captures provenance for a range of ML models across widely used benchmark datasets.
Interesting directions for future work include extending $\syst$ to 
\textcolor{R4}{
(i) identify the influence of groups of instances, e.g., subpopulations sharing a feature value or label, by applying provenance summarization techniques~\cite{LL20e,LL20,AL23} over $\globalProv$;
} 
(ii) support deeper recursive dependencies in more complex ML algorithms such as neural networks, and (iii) compute new types of explanations using captured provenance and explanation data to quantify marginal feature contributions.

\enlargethispage{\baselineskip}



\bibliographystyle{ACM-Reference-Format}
\bibliography{sample}

\end{document}